\documentclass[12pt]{article} 
\usepackage{graphicx,color}
\usepackage{url}
\usepackage{amsmath,amssymb,amsfonts}

\newtheorem{theorem}{Theorem}
\newtheorem{lemma}[theorem]{Lemma}

\def\shadowbox{\hbox{\rule[-0.0ex]{0.1ex}{1.2ex}%
\hspace{-0.1ex}\rule[-0.0ex]{1.2ex}{0.1ex}%
\hspace{0.0ex}\rule[-0.0ex]{0.1ex}{1.2ex}\hspace{-1.3ex}%
\rule[1.15ex]{1.25ex}{0.1ex}\hspace{-0.0ex}\rule[-0.25ex]{0.3ex}{1.1ex}%
\hspace{-1.2ex}\rule[-0.25ex]{1.1ex}{0.25ex}}}
\def\qed{\ifmmode \hbox{\hfill\shadowbox}
     \else \hphantom{x}\hfill\shadowbox \fi}
\def\QED{\mbox{\phantom{m}}\nolinebreak\hfill$\,\shadowbox$}

\def\va{{\bf a}} \def\vb{{\bf b}}

  \def\vs{{\bf s}} 
 \def\vv{{\bf v}} \def\vw{{\bf w}} \def\vx{{\bf x}} 
\def\vy{{\bf y}} \def\vz{{\bf z}}

\def\vA{{\bf A}} \def\vB{{\bf B}}  \def\vD{{\bf D}}
   
\def\vI{{\bf I}}   
\def\vM{{\bf M}}   
  \def\vS{{\bf S}} \def\vT{{\bf T}}
\def\vU{{\bf U}} \def\vV{{\bf V}}   
 \def\vZ{{\bf Z}}
\def\vPsi{{\bf \Psi}}
\def\trace{{\operatorname{trace}}}

\def\vUl{{\vU^{(l)}}}
\def\vVk{{\vV^{(k)}}}

\def\eps{\varepsilon}
\def\CC{{\mathbb C}}
\def\NN{{\mathbb N}}

\def\RR{{\mathbb R}}

\def\Prob{{\mathbb P}}

\def\Exp{{\mathbb E}}

\def\Real{\operatorname{Re}}

\def\sgn{\operatorname{sgn}}
\def\supp{\operatorname{supp}}

\def\dopp{{f}}
\def\aR{{\va_R(\beta)}}
\def\aRp{{\va_R(\beta')}}

\def\aT{{\va_T(\beta)}}
\def\aTp{{\va_T(\beta')}}

\def\St{{\vS_{\tau}}}
\def\Stp{{\vS_{\tau'}}}
\def\Std{{\vS_{\tau,\dopp}}}
\def\Stdp{{\vS_{\tau',\dopp'}}}
\def\colA{{\vA_{\tau,\beta}}}
\def\colAp{{\vA_{\tau',\beta'}}}

\def\Dt{{\Delta_\beta}}
\def\Dtau{{\Delta_\tau}}

\def\op{{\text{op}}}
\def\conj#1{{\overline#1}}
\def\nfrom{{-\frac{N_R N_T}{2}}}
\def\nto{{\frac{N_R N_T}{2}-1}}
\def\vmat{{\bf \Psi}}
\def\vAt{\tilde{\vA}}
\def\vxt{{\vz}}
\def\vTs{{\vT}_{N_t}}
\def\vIT{{\vI}_{N_T}}
\def\Kmax{{K_{\max}}}
\def\SNRmin{{\text{SNR}_{\min}}}

\date{}

\title{Analysis of Sparse MIMO Radar}

\author{\begin{tabular}[t]{c@{\extracolsep{8em}}c} 
Thomas Strohmer\thanks{T.S. was supported by the National Science
Foundation under grant DMS-0811169 and by DARPA under grant N66001-11-1-4090.}
  & Benjamin Friedlander \thanks{B.F. was supported by the National Science 
Foundation under grant CCF-0725366.} \\
Dept. Mathematics & Dept. Elec. Eng. \\
UC Davis & UC Santa Cruz \\
Davis CA 95616 & Santa Cruz, CA 95064
\end{tabular}}

\begin{document}
\maketitle

\begin{abstract}
We consider a multiple-input-multiple-output radar system and derive a 
theoretical framework for the recoverability of targets in the 
azimuth-range domain and the azimuth-range-Doppler domain
via sparse approximation algorithms. Using tools developed in the area of
compressive sensing, we prove bounds on the number of detectable targets 
and the achievable resolution in the presence of additive noise. 
Our theoretical findings are validated by numerical simulations.

\vspace{.2in}

\end{abstract}

\section{Introduction} \label{s:intro}

While radar systems have been in use for many decades, radar is far from
being a `solved problem'. Indeed, exciting new developments in radar
pose great challenges both to engineers and mathematicians~\cite{CB09}. 
Two such
developments are the advent of MIMO (multi-input multi-output) radar
\cite{mimobook}, and the application of compressed sensing to radar
signal processing \cite{HS09}.

MIMO radar is characterized by using multiple antennas to simultaneously
transmit diverse, usually orthogonal, waveforms in addition to using
multiple antennas to receive the reflected signals. MIMO radar has the
potential for enhancing spatial resolution and improving interference
and jamming suppression. The ability of MIMO radar to shape the transmit
beam post facto allows for adapting the transmission based on the
received data in a way which is not possible in non-MIMO radar.

A radar system illuminates a given area and attempts to detect and
determine the location of objects of interest in its field of view, and
to estimate their strength (radar reflectivity). The space of interest
may be divided into range-azimuth (distance and direction) cells, or
range-Doppler-azimuth (distance, direction and speed) cells in the case
there is relative motion between the radar and the object. In many cases
the radar scene is sparse in the sense that only a small fraction (often
a {\sl very} small fraction) of the cells is occupied by the objects of
interest. 

Conventional radar processing does not take into account the a-priori
knowledge that the radar scene is sparse. Recent works, such as
\cite{HS09,PEPC2010} developed techniques which attempt to exploit 
this sparsity 
using tools from the area of compressed sensing~\cite{CanTao,Don}. The
exploitation of sparsity has the potential to improve the performance of
radar systems under certain conditions and is therefore of considerable
practical interest.

In this paper we study the issue of sparsity in the specific context of
a MIMO radar system employing multiple antennas at the transmitter the
receiver, where the two arrays are co-located. We note that related work
on the application of compressive sensing techniques to MIMO radar can be
found in \cite{YAP11,YAP11a}. Our emphasis here is on developing
the basic theory needed to apply sparse recovery techniques for the
detection of the locations and reflectivities of targets for MIMO
radar.  

The basic model for the problem we are considering involves a linear
measurement equation $\vy = \vA \vx +\vw$ where $\vy$ is a vector of
measurements collected by the receiver antennas over an observation
interval, $\vA$ is a measurement matrix whose columns correspond to the
signal received from a single unit-strength scatterer at a particular
range-azimuth (or range-azimuth-Doppler) cell, $\vx$ is a vector whose
elements represent the complex amplitudes of the scatterers, and $\vw$
is a noise vector. The measurement equation is assumed to be
under-determined, possibly highly under-determined. The sparsity
of the radar scene is introduced by assuming that only $K$ elements of
the vector $\vx$ are non-zero, where $K$ is much smaller than the
dimension of the vector. The measurement matrix $\vA$ embodies in it the
details of the radar system such as the transmitted waveforms
and the structure of antenna array.

In this paper we study the conditions under which this problem has a
satisfactory solution. This is a fundamental issue of both theoretical
and practical importance. More specifically, the analysis presented in
the following sections addresses the following issues:

\begin{itemize}

\item It is known from the theory of compressed sensing
\cite{CanTao,Don} that
the matrix $\vA$ must satisfy certain conditions in order that the
solution computed via an appropriate convex program will indeed coincide
with the desired sparsest solution (whose computation is in general an
NP-hard problem). In our problem the characteristics of this matrix
depend on the choice of the radar waveforms and the number and positions
of the transmit and receive antennas. We develop the results necessary
for understanding how the selection of the parameters of the radar
system affects the conditions mentioned above.

\item The ability of the algorithm to correctly detect targets depends
on the number of these targets, $K$, and the signal to noise ratio. We
show that as long as the number of the targets is less than a maximal
value $K_{\max}$, and the signal to noise is larger than some minimal
value ${\rm SNR}_{\min}$, the targets can be correctly detected with high
probability by solving an $\ell_1$-regularized least squares problem
known under the name {\em lasso}. Explicit
formulas are presented for $K_{\max}$ and ${\rm SNR}_{\min}$ as a
function of the number of transmit and receive antennas and  the number
of azimuth and range cells.

\end{itemize}

The structure of the paper is as follows. Subsection~\ref{ss:notation}
introduces notation used throughout the paper. In Section~\ref{s:setup}
we describe the problem formulation and the setup. We derive conditions
for the recovery of targets in the Doppler-free case in
Section~\ref{s:nodoppler}, and the case of detecting targets in presence of
Doppler is analyzed in Section~\ref{s:doppler}. Our theoretical results are
supported by numerical simulations, see Section~\ref{s:numerics}. We
conclude in Section~\ref{s:conclusion}. Finally, some auxiliary results are
collected in the appendices.

\subsection{Notation} \label{ss:notation}

Let $\vv \in \CC^n$. As usual, we define $\|\vv\|_1 : = \sum_{k=1}^{n}
|\vv_k|$ and $\|\vv\|_2:= \sqrt{\sum_{k=1}^{n} |\vv_k|^2}$.
For a given matrix $\vA$ we denote its $k$-th column by $\vA_k$ and
the element in the $i$-th row and $k$-th column by $\vA_{[i,k]}$.
The operator norm of $\vA$ is the largest singular value of
$\vA$ and is denoted by $\|\vA\|_{\op}$, the Frobenius norm of $\vA$
is $\|\vA\|_F = \sqrt{\sum_{i,k}|\vA_{[i,k]}|^2}$.
The coherence of $\vA$ is defined as
\begin{equation}
\mu(\vA) := \underset{k \neq l}{\max} \,\,
\frac{|\langle \vA_k, \vA_l \rangle|}{\|\vA_k\|_2 \|\vA_l\|_2}.
\label{eq:coherence}
\end{equation}
For $\vx \in \CC^n$, let
$\vT_\tau$ denote the circulant translation operator, defined by
\begin{equation}
\label{translation}
\vT_\tau \vx(l) = \vx(l-\tau),
\end{equation}
where $l-\tau$ is understood modulo $n$,
and let $\vM_f$ be the modulation operator defined by
\begin{equation}
\label{modulation}
\vM_f \vx(l) = \vx(l) e^{2\pi i lf}.
\end{equation}

\section{Problem formulation and signal model}
\label{s:setup}

We refer to~\cite{Rihk,CB09} for the mathematical foundations of radar 
and to~\cite{LS09} for an introduction to MIMO radar. However, the reader 
needs only a very basic knowledge of the mathematical concepts underlying 
radar to be able to follow our approach.

We consider a MIMO radar employing $N_T$ antennas at the transmitter and
$N_R$ antennas at the receiver. We assume that the element spacing is
sufficiently small so that the radar return from a given scatterer is
fully correlated across the array. In other words, this is a coherent
propagation scenario. 

To simplify the presentation we assume that the two arrays are
co-located, i.e. this is a mono-static radar. The extension to the
bi-static case is straightforward as long as the coherency assumption
holds for each array. The arrays are characterized by the array
manifolds: $\va_R(\beta)$ for the receive array and $\va_T(\beta)$ for
the transmit array, where $\beta = \sin(\theta)$ is the direction relative 
to the array. We assume that the arrays and all the scatterers are in the same
2-D plane. The extension to the 3-D case is straightforward and all of
the following results hold for that case as well.

For convenience we formulate our theorems and analysis in terms of delay 
$\tau$ instead of range $r$. This is no loss of generality, as delay and 
range are related by $\tau = 2 r/c$, with $c$ denoting the speed of light.

\subsection{The model for the azimuth-delay domain} \label{ss:nodoppler}

The $i$-th transmit antenna repeatedly transmits the signal $s_i(t)$.
Let $\vZ(t; \beta, \tau)$ be the $N_R \times N_t$ noise-free received signal 
matrix from a unit strength target at direction $\beta$ and delay $\tau$, 
where $N_t$ is the number of samples in time. Then
\begin{equation}
\vZ(t; \beta, \tau) = \va_R(\beta) \va_T^T(\beta) \vS^{T}_\tau,
\notag
\end{equation}
where $\St$ is an $N_t \times N_T$ matrix whose columns are the 
circularly delayed signals $s_i(t-\tau)$, sampled at the discrete time
points $t=n \Delta_t, n=1,\dots,N_t$. 
If $\tau =0$, we often write simply $\vS$ instead of $\vS_0$.

Assuming uniformly spaced linear arrays, the array manifolds are given
by
\begin{equation}
\label{transmit_array}
\va_T(\beta) = \left[\begin{array}{c}
1 \\
e^{j 2 \pi d_T \beta} \\
\vdots \\
e^{j 2 \pi d_T \beta (N_T-1)} \\
\end{array}\right]
\end{equation}
and
\begin{equation}
\label{receive_array}
\va_R(\beta) = \left[\begin{array}{c}
1 \\
e^{j 2 \pi d_R \beta} \\
\vdots \\
e^{j 2 \pi d_R \beta (N_R-1)} \\
\end{array}\right]
\end{equation}
where $d_T$ and $d_R$ are the normalized spacings (distance divided by
wavelength) between the elements of the transmit and receive arrays,
respectively. 

The spatial characteristics of a MIMO radar are closely
related to that of a virtual array with $N_T N_R$ antennas, whose array
manifold is $\va(\beta) = \va_T(\beta) \otimes \va_R(\beta)$. It is
known \cite{simo} that the following choices for the spacing of the
transmit and receive array spacing will yield a uniformly spaced virtual
array with half wavelength spacing:
\begin{eqnarray}
d_R = 0.5, d_T = 0.5 N_R; \label{spacing1} \\
d_T = 0.5, d_R = 0.5 N_T. \label{spacing2}
\notag
\end{eqnarray}
Both of these choices lead to a virtual array whose aperture is $0.5
(N_T N_R -1)$ wavelengths. This is the largest virtual aperture free of
grating lobes. The choices~\eqref{spacing1} and~\eqref{spacing2} will
also show up in our theoretical analysis, e.g.\ see Theorem~\ref{th:lasso}.

Next let $\vz(t; \beta, \tau) = {\rm vec}\{\vZ\}(t; \beta, \tau)$ be the noise-free
vectorized received signal. We set up a discrete delay-azimuth grid
$\{(\beta_i,\tau_j)\}, 1\le i\le N_\beta, 1\le j\le N_\tau$, where
$\Delta_\beta$ and $\Delta_\tau$ denote the corresponding discretization
stepsizes. Using vectors $\vz(t; \beta_i,\tau_j)$ for all grid points 
$(\beta_i,\tau_j)$ we construct a complete response matrix $\vA$ whose columns 
are $\vz(t; \beta_i,\tau_j)$ for $1 \le i \le  N_{\beta}$ and 
$1 \le j \le  N_\tau$. In
other words, we have   $N_\tau$ delay values and $N_{\beta}$ azimuth
values, so that $\vA$ is a $N_R N_t \times N_\tau N_{\beta}$ matrix.

Assume that the radar illuminates a scene consisting of $K$ scatterers
located on $K$ points of the $(\beta,\tau_j)$ grid. Let $\vx$ be a
sparse vector whose non-zero elements are the complex amplitudes of the
scatterers in the scene. The zero elements corresponds to grid points
which are not occupied by scatterers. We can then define the radar
signal $\vy$ received from this scene by
\begin{equation}
\label{radarsystem}
\vy = \vA \vx + \vv
\end{equation}
where $\vy$ is a $N_R N_t\times 1$  vector,  $\vx$ is a  $N_\tau N_{\beta}
\times 1 $ sparse vector, $\vv$ is a $N_R N_t\times 1$ complex Gaussian
noise vector, and  $\vA$ is a $N_R N_t\times N_\tau N_{\beta}$ matrix.

\subsection{The model for the azimuth-delay-Doppler domain} \label{ss:doppler}

The discussion so far was for the case of a stationary radar scene and a
fixed radar, in which case there is no Doppler shift. The extension of
this signal model to include the Doppler effect is conceptually
straightforward, but leads to a significant increase in the problem
dimension.

The signal model for the return from a unit strength scatterer at
direction $\beta$, delay $\tau$, and Doppler $f$ (corresponding to its
radial velocity with respect to the radar) is given by
\begin{equation}
\vZ(t; \beta, \tau,f) = \va_R(\beta) \va_T^T(\beta) \vS^{T}_{\tau,f},
\notag
\end{equation}
where $\Std$ is a $N_t \times N_T$ matrix whose columns are the 
circularly
delayed and Doppler shifted signals $s_i(t-\tau)e^{j 2 \pi f t}$.

As before we let $\vz(t; \beta, \tau, f) = {\rm vec}\{\vZ\}(t; \beta, \tau, f)$
be the noise-free vectorized received signal. We extend the discrete
delay-azimuth grid by adding a discretized Doppler component
(with stepsize $\Delta_f$ and corresponding Doppler values $f=k\Delta_f,
k=1,\dots,N_f$) and
obtain a uniform delay-azimuth-Doppler grid $\{(\beta_i,\tau_j,f_k)\}$.
Using vectors $\vz(t; \beta_i,\tau_j,f_k)$ for all discrete 
$(\beta_i,\tau_j,f_{k})$ 
we construct a complete response matrix $\vA$ whose columns are 
$\vz(t; \beta_i,\tau_j,f_k)$ for $1 \le i \le  N_{\beta}$, 
$1 \le j \le  N_\tau$, $1 \le k \le N_f$. 

Assume that the radar illuminates a scene consisting of $K$ scatterers
located on $K$ points of the $(\beta,\tau_j,f_k)$ grid.  Let $\vx$
be a sparse vector whose non-zero elements are the complex amplitudes of the scatterers
in the scene. The zero elements corresponds to grid points which are not
occupied by scatterers. We can then define the radar signal received
from this scene $\vy$ by

\begin{equation}
\label{noisysystem}
\vy = \vA \vx + \vv
\end{equation}
where $\vy$ is a $N_R N_t\times 1$  vector,  $\vx$ is a  $N_\tau N_{\beta}
N_f \times 1 $ sparse vector, $\vv$ is a $N_R N_t\times 1$ complex Gaussian
noise vector, and  $\vA$ is a $N_R N_t\times N_\tau N_{\beta} N_f$ matrix.

\subsection{The target model}

We define the {\em sign function} for a vector $z\in \CC^n$ as
\begin{equation}
\sgn(z_k) = 
\begin{cases}
z_k/|z_k| & \text{if $z_k \neq 0$,} \\
0         & \text{else.}
\end{cases}
\label{sgn}
\end{equation}

We introduce the following {\em generic $K$-sparse target model}:
\begin{itemize}
\item The support $I_K \subset \{1,\dots,N_\tau N_\beta\}$ of the $K$ nonzero
coefficients of $\vx$ is selected uniformly at random.
\item The non-zero coefficients of $\sgn(\vx)$ form a Steinhaus sequence, 
i.e., the phases of the non-zero entries of $\vx$ are random and uniformly 
distributed in $[0,2\pi)$.
\end{itemize}
We do not impose any condition on the amplitudes of the non-zero entries
of $\vx$. We do assume however that the targets are exactly located
at the discretized grid points. This is certainly an idealized assumption,
that is not satisfied in this strict sense in practice, resulting in
a ``gridding error''. We refer the reader to~\cite{HS10,CPS10} for an
initial analysis of the associated perturbation error, and to~\cite{FW12} 
for an interesting numerical approach to deal with this issue.

\subsection{The recovery algorithm -- Debiased Lasso}

A standard approach to find a sparse (and under appropriate
conditions {\em the sparsest}) solution to a noisy system
$\vy = \vA\vx + \vw$ is via
\begin{equation}
\underset{\vx}{\min}\, \frac{1}{2}\|\vA\vx - \vy\|_2^2 + \lambda \|\vx\|_1,
\label{lassoD}
\end{equation}
which is also known as lasso~\cite{Tib96}.
Here $\lambda >0$ is a regularization parameter.

In this paper we adopt the following two-step version of lasso. 
In the first step we compute an estimate $\tilde{I}$ for the support of $\vx$ 
by solving~\eqref{lassoD}. In the second step we estimate the amplitudes 
of $\vx$ by solving the reduced-size least squares problem 
$\min \|\vA_{\tilde{I}} \vx_{\tilde{I}} - \vy\|_2$, where $\vA_{\tilde{I}}$ 
is the submatrix of $\vA$ consisting of the columns corresponding
to the index set $\tilde{I}$, and similarly for $\vx_{\tilde{I}}$.
This is a standard way to ``debias'' the solution, we thus will call 
this approach in the sequel {\em debiased lasso}.


\section{Recovery of targets in the Doppler-free case}
\label{s:nodoppler}

We assume that $s_i(t)$ is a periodic, continuous-time white Gaussian
noise signal of period-duration $T$ seconds and bandwidth $B$. The transmit
waveforms are normalized so that the total transmit power is fixed,
independent of the number of transmit antennas. Thus, we assume that
the entries of $s_i(t)$ have variance $\frac{1}{N_T}$.
It is convenient to introduce the
finite-length vector $\vs_i$ associated with $s_i$, via
$\vs_i(l) := s_i(l \Delta_t), l = 1,\dots,N_t$, where $\Delta_t = \frac{1}{2B}$
and $N_t = T/\Delta_t$.


\begin{theorem}
\label{th:lasso}
Consider $\vy = \vA \vx +\vw$, where $\vA$ is as defined in
Subsection~\ref{ss:nodoppler} and $\vw_i \in {\cal CN}(0,\sigma^2)$.
Choose the discretization stepsizes to be $\Delta_\beta = \frac{2}{N_R N_T}$ 
and $\Delta_\tau = \frac{1}{2B}$. Let $d_T = 1/2, d_R = N_T/2$ or 
$d_T = N_R/2, d_R = 1/2$, and suppose that 
\begin{equation}
N_t \ge 128, \qquad N_\tau \ge \sqrt{N_\beta}, \qquad \text{and} \qquad
\big(\log (N_\tau N_\beta) \big)^3 \le N_t.
\label{coherenceproperty2}
\end{equation}
If $\vx$ is drawn from the generic $K$-sparse target model with
\begin{equation}
K \le \Kmax := \frac{c_0 N_\tau N_R}{3N_T \log (N_\tau N_\beta)}
\label{lassosparsity1}
\end{equation}
for some constant $c_0>0$, and if
\begin{equation}
\underset{k \in I}{\min}\, |\vx_k| > 
                   \frac{10\sigma}{\sqrt{N_R N_t}} \sqrt{2 \log N_\tau N_\beta},
\label{amplitudeproperty2}
\end{equation}
then the solution $\tilde{\vx}$ of the debiased lasso computed with
$\lambda = 2 \sigma \sqrt{2 \log(N_\tau N_\beta)}$
obeys
\begin{equation}
\label{support2}
\supp (\tilde{\vx}) = \supp (\vx),
\end{equation}
with probability at least 
\begin{equation}
\notag
(1 - p_1)(1 - p_2)(1-p_3)(1-p_4),
\end{equation}
and
\begin{equation}
\frac{\|\tilde{\vx} - \vx \|_2}{\|\vx\|_2} 
    \le \frac{\sigma \sqrt{12 N_t N_R}}{\|\vy\|_2}
\label{error2}
\end{equation}
with probability at least 
\begin{equation}
\notag
(1 - p_1)(1 - p_2)(1-p_3)(1-p_4)(1 - p_5),
\end{equation}
where
$$p_1 = e^{-\frac{(1-\sqrt{1/3})^2 N_t}{2}} + N_t^{1 - C N_T},$$
$$p_2=2e^{-\frac{N_t(\sqrt{2}-1)^2}{4}}+2 (N_R N_T)^{-1}-6(N_t N_\beta)^{-1},$$
$$p_3 = e^{-\frac{(1-\sqrt{1/3})^2 N_t}{2}}, \quad
p_4 = N_R N_T e^{-\frac{N_R N_t}{25}},$$
and
$$
p_5 = 2(N_\tau N_\beta)^{-1}(2\pi \log (N_\tau N_\beta) + K(N_\tau N_\beta)^{-1})
+ {\cal O}((N_\tau N_\beta)^{-2 \log 2}).$$
\end{theorem}

\medskip
\noindent
{\bf Remark:}
\begin{itemize}
\item[(i)] 
While the expressions for the probability of success in the
above theorem are admittedly somewhat unpleasant, we point out that
the individual terms are fairly small. Moreover, the probabilities
can easily be made smaller by slightly increasing the constants
in the assumptions on $N_t,N_R,N_T$.

\item[(ii)] The assumptions in~\eqref{coherenceproperty2}
are fairly mild and easy to satisfy in practice.

\item[(iii)] We emphasize that there is no constraint on the dynamic
range of the target amplitudes. The lasso estimate will recover
all target locations correctly as long as they exceed the noise
level~\eqref{amplitudeproperty2}, regardless of the dynamical range 
between the targets. 

\item[(iv)] We note that $|x_k|^2 / \sigma^2$ is the signal-to-noise
ratio for the $k$-th scatterer at the receiver array input. The
measurement vector $\vy$ provides $N_R N_t$ measurements of $x_k$.
Therefore it is useful to define the signal-to-noise ratio  associated
with the $k$-th scatterer as  $\text{SNR}_k = N_R N_t |x_k|^2 /
\sigma^2$. This is often referred to as the output SNR because it is the
effective SNR at the output of a matched-filter receiver.
Equation~\eqref{amplitudeproperty2} can thus be written as
$\text{SNR}_k > 200 \log N_\tau N_\beta$,
However, the factor 200 is definitely
way too conservative. As is evident from the comments following Theorem 1.3 
in~\cite{CP08}, one can replace the factor 10 in~\eqref{amplitudeproperty2} 
by a factor $(1+\eps)$ for some $\eps>0$, at the cost of a somewhat reduced 
probability of success and some slightly stronger conditions on the coherence 
and sparsity. This indicates that the SNR condition for which perfect
target detection can be achieved is
\begin{equation}
\label{SNRbound}
\text{SNR} \ge \SNRmin := C \log N_\tau N_\beta ,
\end{equation}
where $C$ is a constant of size ${\cal O}(1)$.

\item[(v)] The condition that the target locations are assumed to be random
can likely be removed by using a different proof technique that relies
on a dual certificate approach (e.g. see~\cite{candes_plan2011}) and 
tools developed in~\cite{Rau10}. We do not pursue this direction in this paper.
\end{itemize}

The proof of Theorem~\ref{th:lasso} is carried out in several steps. We
need two key estimates, one concerns a bound for the operator norm of $\vA$, 
the other one concerns a bound for the coherence of $\vA$. 
We start with deriving a bound for $\|\vA\|_{\op}$.

\begin{lemma}
\label{th:normbound}
Let $\vA$ be as defined in Theorem~\ref{th:lasso}. Then
\begin{equation}
\label{A_estimate}
\Prob \Big( \|\vA\|^2_{\op} \ge N_t N_R N_T(1+\log N_t)\Big)
\le N_t^{1 - C N_T},
\end{equation}
where $C>0$ is some numerical constant.
\end{lemma}

\begin{proof}
There holds $\|\vA\|_\op^2 = \|\vA \vA^{\ast}\|_\op$.
It is convenient to consider $\vA \vA^{\ast}$ as block matrix
$$
\begin{bmatrix}
\vB_{1,1}            & \vB_{1,2}   & \dots      & \vB_{1,N_R} \\
\vdots             & \ddots    &            & \vdots    \\
\vB_{N_R,1}^{\ast}   &           &            & \vB_{N_R,N_R}
\end{bmatrix},
$$
where the blocks $\{\vB_{i,i'}\}_{i,i'=1}^{N_R}$ are matrices of size 
$N_t \times N_t$. We claim that $\vA \vA^{\ast}$ is a block-Toeplitz matrix 
(i.e., $\vB_{i,i'} = \vB_{i+1,i'+1}, i=1,\dots, N_R-1$) and the individual 
blocks $\vB_{i,i'}$ are circulant matrices. To see this, recall the structure 
of $\vA$ and consider the entry $\vB_{[i,l;i',l']}$, $i,i'=1,\dots,N_R;
l,l'=1,\dots,N_t$:
\begin{gather}
\vB_{[i,l;i',l']}  =  (\vA \vA^{\ast})_{[i,l;i',l']} = 
\sum_{\beta} \sum_{\tau} \vA_{[i,l;\beta,\tau]} 
\vA_{[i',l';\beta,\tau]} \notag \\
 = \sum_{\beta} \sum_{n=1}^{N_\tau} \aR_i 
\sum_{k=1}^{N_T} \aT_k s_k(l \Delta_t - n \Delta_\tau) \conj{\aR}_{i'} 
\overline{\sum_{k'=1}^{N_T} \aT_{k'} s_{k'}(l' \Delta_t - n \Delta_\tau)}
\notag \\
 =  \sum_{\beta}  \aR_i \conj{\aR}_{i'} 
\sum_{k=1}^{N_T} \sum_{k'=1}^{N_T} \aT_k \conj{\aT}_{k'} 
\sum_{n=1}^{N_\tau} s_k(l \Delta_t  - n \Delta_\tau) 
\overline{s_{k'}(l' \Delta_t  - n \Delta_\tau )} \notag \\
 =  \sum_{\beta} e^{j2\pi d_R (i-i') \beta}
\sum_{k=1}^{N_T} \sum_{k'=1}^{N_T} e^{j2\pi d_T (k-k')\beta}
\sum_{n=1}^{N_\tau} s_k(l\Delta_t-n \Delta_\tau) 
\overline{s_{k'}(l' \Delta_t-n \Delta_\tau)},
\label{matrixstructure}
\end{gather}
where we used the delay discretization $\tau = n\Dtau, n= 1,\dots,N_\tau$.
The block-Toeplitz structure, $\vB_{i,i'} = \vB_{i+1,i'+1}$, follows
from observing that the expression~\eqref{matrixstructure} depends
on the difference $i-i'$, but not on the individual values of $i,i'$.
The circulant structure of an individual block $\vB_{i,i'}$ ($i,i'$
are now fixed) follows readily from noting that 
$$\sum_{n=1}^{N_\tau} s_k(l \Delta_t- n \Delta_\tau)
\overline{s_{k'}(l'\Delta_t- n \Delta_\tau)}
=\sum_{n=1}^{N_\tau} s_k((l+1)\Delta_t-n \Delta_\tau) 
  \overline{s_{k'}((l'+1) \Delta_t - n \Delta_\tau)},$$
since we have chosen $\Delta_t = \Delta_\tau$ 
and since the shifts are circulant in this case.

We will now show that the blocks $B_{i,i'}$ are actually
zero-matrices for $i \neq i'$. For convenience we introduce the notation
\begin{equation}
\notag
G_{k,k'}(l,l'):=
\sum_{n=1}^{N_\tau}s_k(l\Delta_t -n \Dtau)\overline{s_{k'}(l' \Delta_t-n\Dtau)},
\qquad l,l' = 1,\dots,N_t; k,k' = 1,\dots,N_T,
\end{equation}
Substituting $d_T = 1/2, d_R = N_T/2$ (the very similar calculation for
$d_R = 1/2, d_T = N_R/2$ is left to the reader) and the discretization 
$\beta = n\Delta_\beta, n=1,\dots,N_\beta,$ 
with $\Delta_\beta=\frac{2}{N_R N_T}$ in~\eqref{matrixstructure} we can write
\begin{gather}
\vB_{[i,l;i',l']} = 
\sum_{n=\nfrom}^{\nto} e^{j2\pi \frac{N_T}{2} (i-i') \frac{2n}{N_R N_T}}
\sum_{k=1}^{N_T} \sum_{k'=1}^{N_T} 
e^{j2\pi \frac{1}{2}(k-k')\frac{2n}{N_R N_T}} G_{k,k'}(l,l') \notag\\
= \sum_{k=1}^{N_T} \sum_{k'=1}^{N_T} G_{k,k'}(l,l')
\sum_{n=0}^{N_R N_T-1} e^{j2\pi N_T (i-i') \frac{n}{N_R N_T}}
e^{j2\pi(k-k')\frac{n}{N_R N_T}}. 
\label{Bsum}
\end{gather}
We analyze the inner summation in~\eqref{Bsum} separately. 
\allowdisplaybreaks
\begin{gather}
\sum_{n=0}^{N_R N_T-1} e^{j2\pi N_T (i-i') \frac{n}{N_R N_T}}
e^{j2\pi(k-k')\frac{n}{N_R N_T}} =  
\sum_{n_1=0}^{N_T-1} \sum_{n_2=0}^{N_R-1}
e^{j2\pi(k-k')\frac{n_1 N_R + n_2}{N_R N_T}} 
 e^{j2\pi N_T (i-i') \frac{n_1 N_R + n_2}{N_R N_T}}\notag \\
= \sum_{n_2=0}^{N_R-1} e^{j2\pi(k-k')\frac{n_2}{N_R N_T}} 
 e^{j2\pi (i-i') \frac{n_2 N_T}{N_R N_T}} 
\sum_{n_1=0}^{N_T-1} e^{j2\pi (k-k') \frac{n_1 N_R}{N_R N_T}} 
 e^{j2\pi (i-i') \frac{n_1 N_R N_T}{N_R N_T}} \notag\\
= \sum_{n_2=0}^{N_R-1} e^{j2\pi(k-k')\frac{n_2}{N_R N_T}} 
 e^{j2\pi (i-i') \frac{n_2}{N_R}} 
\sum_{n_1=0}^{N_T-1} e^{j2\pi (k-k') \frac{n_1}{N_T}} 
 \underbrace{e^{j2\pi (i-i') n_1}}_{\text{$=1$ for all $i,i'$}} \notag\\
= \sum_{n_2=0}^{N_R-1} e^{j2\pi(k-k')\frac{n_2}{N_R N_T}} 
 e^{j2\pi (i-i') \frac{n_2}{N_R}} 
\sum_{n_1=0}^{N_T-1} e^{j2\pi (k-k') \frac{n_1}{N_T}} \notag\\
= \sum_{n_2=0}^{N_R-1} e^{j2\pi(k-k')\frac{n_2}{N_R N_T}}
 e^{j2\pi (i-i') \frac{n_2}{N_R}} N_T \delta_{k-k'}.\notag
\end{gather}
Hence
\begin{gather*}
\vB_{[i,l;i',l']} 
= N_T \sum_{k=1}^{N_T} \sum_{k'=1}^{N_T} \delta_{k-k'} G_{k,k'}(l,l')
\sum_{n_2=0}^{N_R-1} 
\underbrace{e^{j2\pi(k-k')\frac{n_2}{N_R N_T}}}_{\text{$=1$ for $k=k'$}}
 e^{j2\pi (i-i') \frac{n_2}{N_R}}  \\
= N_T \sum_{k=1}^{N_T} G_{k,k}(l,l')
\sum_{n_2=0}^{N_R-1} e^{j2\pi (i-i') \frac{n_2}{N_R}}  
= N_T N_R \sum_{k=1}^{N_T} G_{k,k}(l,l') \delta_{i-i'}.
\end{gather*}
Thus, $\vB_{i,i'} = {\bf 0}$ for $i\neq i'$, and $\vA^{\ast}\vA$ is indeed a
block-diagonal matrix, which in turn implies
$\|\vA\|_\op^2 = \max_{i} \|\vB_{i,i}\|_\op$. But due
to the block-Toeplitz structure of $\vA^{\ast}\vA$ we have
$\vB_{1,1}=\vB_{2,2}=\dots = \vB_{N_R,N_R}$. Therefore 
\begin{equation}
\|\vA\|_\op^2 = \|\vB_{1,1}\|_\op.
\label{AB_estimate}
\end{equation}

To bound $\|\vB_{1,1}\|_\op$ we utilize its circulant structure
as well as tail bounds of quadratic forms.
Let $\vb$ be the first column of $\vB_{1,1}$, then 
$\|\vB_{1,1}\|_\op =\sqrt{N_t} \|\hat{\vb}\|_{\infty}$ where $\hat{\vb}$
is the Fourier transform of $\vb$. From our previous computations we
have (after a change of variables)
\begin{equation}
\vb(l) = N_T N_R \sum_{k=1}^{N_T} G_{k,k}(l,0)
= N_T N_R \sum_{k=1}^{N_T} \sum_{n=1}^{N_\tau} s_k(n \Dtau-l
\Delta_t)\overline{s_k(n \Dtau)},
\qquad l=0,\dots,N_t-1.
\notag
\end{equation}
We will rewrite this expression so that we can apply
Lemma~\ref{le:quadform} to bound $\|\hat{\vb}\|_{\infty}$. Let $\vTs$ 
denote the translation operator
on $\CC^{N_t}$ as introduced in~\eqref{translation} and define the
$N_t N_T \times N_t N_T$ block-diagonal matrix $\vUl=\{u^{(l)}_{ii'}\}$ by 
\begin{equation}
\label{eq:quadmatrix}
\vUl: = N_R N_T \sqrt{N_t} \vIT \otimes \vTs^l, 
    \qquad \text{for $l=0,\dots,N_t-1.$}
\end{equation}
Furthermore, let $\vz = [ \vs_1^T,\vs_2^T,\dots,\vs_{N_T}^T]^T$, then
\begin{equation}
\notag
\sqrt{N_t} \vb(l) =\sqrt{N_t} N_T N_R \sum_{k=1}^{N_T} \langle \vs_k, \vTs^l \vs_k \rangle
       = \langle \vz ,\vUl \vz \rangle,
       =  \sum_{i,i'=1}^{N_t N_T} u^{(l)}_{i i'} \bar{\vz}_i \vz_{i'}.
\notag
\end{equation}
and therefore
\begin{equation*}
\sqrt{N_t}\hat{\vb}(k) = \frac{1}{\sqrt{N_t}}
   \sum_{l=0}^{N_t-1}\sum_{i,i'=1}^{N_t N_T} u^{(l)}_{ii'}
   \bar{\vz}_i\vz_{i'} e^{j2\pi kl/N_t} =
   \sum_{i,i'=1}^{N_t N_T} \bar{\vz}_i\vz_{i'} 
   \frac{1}{\sqrt{N_t}}\sum_{l=0}^{N_t-1} u^{(l)}_{ii'} e^{j2\pi kl/N_t}
=  \sum_{i,i'=1}^{N_t N_T} \bar{\vz}_i\vz_{i'} v^{(k)}_{ii'},
\end{equation*}
where we have denoted $v^{(k)}_{ii'}:= 
\frac{1}{\sqrt{N_t}}\sum_{l=0}^{N_t-1} u^{(l)}_{ii'} e^{j2\pi kl/N_t}$ 
for $i,i' = 0,\dots,N_t N_T-1$ and $k=0,\dots,N_t-1$. 
It follows from~\eqref{eq:quadmatrix} and standard properties of
the Fourier transform that the matrix 
$\vVk:=\{v^{(k)}_{ii'}\}$ is a block-diagonal matrix with $N_T$ blocks
of size $N_t \times N_t$, where each non-zero entry of such a block
has absolute value $N_R N_T$. Furthermore, a little algebra shows that
$\|\vVk\|_{F} = \sqrt{N_t^2 N_R^2 N_T^3}$,
$\|\vVk\|_{\text{op}} = N_t N_R N_T$, $\trace(\vVk) = N_t N_R N_T^2$, and
$$\Exp \big(\sum_{i,i'=1}^{N_t N_T} \bar{\vz}_i\vz_{i'} v^{(k)}_{ii'}\big)
 = \frac{1}{N_T}\trace(\vVk) = N_t N_R N_T.$$
We can now apply Lemma~\ref{le:quadform} (keeping in mind that
$x_i \sim {\cal CN}(0,\frac{1}{N_T})$) and obtain 
\begin{equation}
\Prob \big(|\sqrt{N_t} \hat{\vb}(l)| \ge N_t N_R N_T + t \big) \le \exp 
\Big(- C \min\Big\{\frac{t N_T}{N_t N_R N_T},\frac{t^2 N_T^2}{N_t^2 N_R^2 N_T^3}\Big\} 
\Big),
\notag
\end{equation}
where $C>0$ is some numerical constant.

Choosing $t =   N_t N_R N_T \log N_t$ gives
\begin{equation}
\Prob \big(|\sqrt{N_t} \hat{\vb}(l)| \ge  N_t N_R N_T(1+\log N_t) \big) \le 
\exp  (- C N_T\log N_t),
\notag
\end{equation}
for $l=0,\dots,N_t -1$. Forming the union bound over the $N_t$
possibilities for $l$ gives
\begin{equation}
\Prob \big(\max_l \{|\sqrt{N_t}\hat{\vb}(l)|\}\ge N_t N_R N_T(1+\log N_t)\big)
\le \sum_{l=0}^{N_t-1} \exp  (- C \sqrt{N_T} \log N_t) = N_t^{1 - C N_T}.
\label{probb2}
\end{equation}

We recall that $\|\vB_{1,1}\|_{\text{op}} = \max_l |\sqrt{N_t}\hat{\vb}(l)|$,
and substitute~\eqref{probb2} into~\eqref{AB_estimate} to
complete the proof.
\end{proof}

Next we estimate the coherence of $\vA$. Since the columns of $\vA$
do not all have the same norm, we will proceed in two steps.
First we bound the modulus of the inner product of any two columns
of $\vA$ and then use this result to bound the coherence of a properly 
normalized version of $\vA$.
Since the columns of $\vA$ depend on azimuth and delay, we index them
via the double-index $(\tau,\beta)$. Thus the $(\tau,\beta)$-th column of
$\vA$ is $\vA_{\tau,\beta}$.

\begin{lemma}
\label{th:coherencebound}
Let $\vA$ be as defined in Theorem~\ref{th:lasso}. Assume that 
\begin{equation}
\label{eq:coherencecondition}
N_\tau  \ge \sqrt{N_\beta} \qquad \text{and} \qquad
\log(N_\tau N_\beta) \le \frac{N_t}{30},
\end{equation}
then
\begin{equation}
\underset{(\tau,\beta)\neq (\tau',\beta')}{\max} 
\big|\langle  \vA_{\tau,\beta},\vA_{\tau',\beta'} \rangle \big|
  \le 3 N_R \sqrt{N_t \log (N_\tau N_\beta) }
\label{eq:coherencebound}
\end{equation}
with probability at least $1-2 (N_R N_T)^{-1} - 6(N_\tau N_R N_T)^{-1}$.
\end{lemma}

\begin{proof}
We assume $d_T = \frac{1}{2}, d_R = \frac{N_T}{2}$ and leave the case
$d_T = \frac{N_R}{2}, d_R = \frac{1}{2}$ to the reader.
We need to find an upper bound for
\begin{equation}
\notag
\max{|\langle  \vA_{\tau,\beta},\vA_{\tau',\beta'} \rangle |} \qquad
\text{for $(\tau,\beta) \neq (\tau',\beta')$}.
\end{equation}
It follows from the definition of $z(t;\beta,r)$ via a simple calculation
that 
\begin{equation}
\notag
\vA_{\tau,\beta} = \aR \otimes (\St \aT),
\end{equation}
from which we readily compute
\begin{equation}
\langle \colA,\colAp \rangle = \langle \aR, \aRp \rangle
\langle \St \aT , \Stp \aTp \rangle.
\label{innprod}
\end{equation}
We use the discretization $\beta = n \Delta_\beta$, 
$\beta' = n' \Delta_\beta$, where $\Delta_\beta = \frac{2}{N_R N_T}$,
$n,n' = 1,\dots,N_\beta$, with $N_\beta = N_R N_T$, and obtain after 
a standard calculation
\begin{equation}
\label{eq:orth1}
\langle \aR, \aRp \rangle = 
\begin{cases}
N_R & \text{if $n - n' = k N_R$ for $k =0,\dots,N_T-1$,} \\
0 & \text{if $n - n' \neq k N_R$,} \\
\end{cases}
\end{equation}
and
\begin{equation}
\label{eq:orth2}
\langle \aT, \aTp \rangle = 
\begin{cases}
0 & \text{if $n - n' = k N_R$ for $k =1,\dots,N_T-1$,} \\
\langle \aT, \aT \rangle & \text{if $n - n' =0$.}
\end{cases}
\end{equation}
As a consequence of~\eqref{eq:orth1}, concerning $\beta,\beta'$
we only need to focus on the case $n - n' = k N_R$ for $k=1,\dots,N_T-1$.
Moreover, since 
\begin{equation}
\notag
\langle \St \aT , \Stp \aTp \rangle =
\langle \vS_{\tau - \tau'} \aT , \vS \aTp \rangle,\quad
\text{for $\tau,\tau'=0,\dots,N_\tau-1,$}
\end{equation}
and $|\langle \vS_\tau \aT , \aTp \rangle| = 
|\langle \vS_{N_t - \tau} \aT , \aTp \rangle|$, 
we can confine the range of values for $\tau,\tau'$ to $\tau'=0, \tau=0,\dots,N_t/2$.

We split our analysis into three cases, (i) $\beta \neq \beta', \tau =0$, 
(ii) $\beta \neq \beta', \tau \neq 0$, and 
(iii) $\beta = \beta',\tau  \neq 0$.

\noindent
{\bf Case (i) $\beta \neq \beta', \tau =0$:}
We will first find a bound for
$|\langle \aR, \aRp \rangle \langle \aT , \aTp \rangle|$
and then invoke Lemma~\ref{le:concentration3} to obtain a bound for
$|\langle \aR, \aRp \rangle \langle \vS \aT , \vS \aTp \rangle|$.

Based on~\eqref{eq:orth1} and~\eqref{eq:orth2}, to bound
$|\langle \aR, \aRp \rangle \langle \vS \aT , \vS \aTp \rangle|$
we only need to consider those $n,n'$ for which $n-n'$ is not a 
multiple of $N_R$, in which case $\aT$ and $\aTp$ are orthogonal.
We have
\begin{equation}
\label{innprod3}
|\langle \aR, \aRp \rangle \langle \vS \aT , \vS \aTp \rangle| \le
N_R\, | \langle \vS^{\ast}\vS \aT ,\aTp \rangle|.
\end{equation}
By Lemma~\ref{le:concentration3} there holds
\begin{equation}
\label{innprod4}
\Prob\Big( |\langle \vS^{\ast}\vS \aT ,\aTp \rangle| \ge t N_t \Big)
\le 2 \exp \Big(-N_t \frac{t^2}{C_1 + C_2 t})\Big)
\end{equation}
for all $0 < t < 1$, where $C_1 = \frac{4e}{\sqrt{6\pi}}$ and
$C_2 = \sqrt{8} e$.
We choose $t= 3\sqrt{\frac{1}{N_t} \log (N_\tau N_R N_T)}$ 
in~\eqref{innprod4} and get
\begin{equation}
\label{innprod4a}
\Prob\Big( |\langle \vS^{\ast}\vS \aT ,\aTp \rangle| 
  \ge 3 \sqrt{N_t \log(N_\tau N_R N_T)} \Big)
\le 2 \exp \Big(- \frac{9\log(N_\tau N_R N_T)}{C_1 + \frac{3C_2}{\sqrt{N_t}} 
\sqrt{\log(N_\tau N_R N_T)}}\Big).
\end{equation}
We claim that
\begin{equation}
\frac{9\log(N_\tau N_R N_T)}{C_1 + \frac{3C_2}{\sqrt{N_t}} 
\sqrt{\log(N_\tau N_R N_T)}} \ge 2 \log(N_R N_T).
\label{eq:bound10}
\end{equation}
To verify this claim we first note that~\eqref{eq:bound10} is equivalent to
\begin{equation}
9\log N_\tau  \ge \log(N_R N_T) (2C_1 +
\frac{6C_2}{\sqrt{N_t}}\sqrt{\log(N_\tau N_\beta)} - 9).
\notag
\end{equation}
Using both assumptions in~\eqref{eq:coherencecondition}
and the fact that $2C_1 + \frac{6C_2}{\sqrt{30}} - 9 \le \frac{9}{2}$
we obtain
\begin{equation}
9\log N_\tau \ge \log N_\beta (2C_1 + \frac{6C_2}{\sqrt{30}} - 9)
\ge \log N_\beta (2C_1 + \frac{6C_2}{\sqrt{N_t}}\sqrt{\log(N_t N_\beta)} - 9),
\notag
\end{equation}
which establishes~\eqref{eq:bound10}.
Substituting now~\eqref{eq:bound10} into~\eqref{innprod4a} gives
\begin{equation}
\label{innprod4b}
\Prob\Big( |\langle \vS^{\ast}\vS \aT ,\aTp \rangle| 
\ge 3 \sqrt{N_t \log(N_\tau N_R N_T)} \Big)
\le 2 \exp \big(- 2\log(N_R N_T)\big).
\end{equation}
To bound $\max |\langle \vA_{\tau,\beta}, \vA_{\tau,\beta'} \rangle|$ 
we only have to take the union bound over
$N_R N_T$ different possibilities associated with $\beta,\beta'$,
as $\tau=\tau'=0$. Forming now the union bound, and using~\eqref{innprod3}, 
yields
\begin{equation}
\label{eq:case1}
\Prob\Big( |\langle \vA_{\tau,\beta}, \vA_{\tau,\beta'} \rangle| \le 
3N_R \sqrt{N_t\log (N_\tau N_R N_T)} \Big) \ge 1- 2 (N_R N_T)^{-1}.
\end{equation}

\noindent 
{\bf Case (ii) $\beta \neq \beta', \tau \neq 0$:}
We need to consider the case 
$|\langle \St \aT , \vS \aTp \rangle|$ where $\beta=n\Delta_\beta$,
$\beta'=n' \Delta_{\beta}$, with
$n - n' = k N_R$ for $k = 1,\dots, N_T-1$.
Since the entries of $\vS$ are i.i.d.\ Gaussian random variables, it
follows that the entries of $\St \aT$ are i.i.d.\ 
${\cal CN}(0,1)$-distributed, and similar for $\vS \aTp$. Moreover,
the fact that $\langle \aT,\aTp \rangle = 0$ implies that 
$\St \aT$ and $\vS \aTp$ are independent. Consequently, the entries of
 $\sum_{l=0}^{N_t-1} \overline{(\St \aT)_l} (\vS \aTp)_l$ are
jointly independent. Therefore, we can apply Lemma~\ref{le:gaussianinner}
with $t= 3 \sqrt{N_t \log(N_\tau N_R N_T)}$, form the union bound 
over the $N_\tau N_R N_T$ possibilities associated with $\tau$ (we do not 
take advantage of the fact we actually have only $N_\tau-1$ and not $N_\tau$ 
possibilities for $\tau$) and $\beta,\beta'$ (here, we take again into 
account property~\eqref{eq:orth1}), and eventually obtain
\begin{equation}
\Prob\Big( |\langle \colA,\colAp \rangle| \le 
3 N_R \sqrt{N_t \log(N_\tau N_R N_T)} \Big) \ge 1 - 2 (N_\tau N_R N_T)^{-1}. 
\label{eq:case2}
\end{equation}

\noindent
{\bf Case (iii) $\beta = \beta', \tau \neq 0$:} 
We need to find an upper bound for $|\langle \St \aT , \vS \aT \rangle|$
where $\tau = 1,\dots,N_t-1$. Since
Since each of the entries of
$\St \aT$ and of $\vS \aT$ is a sum of $N_T$ i.i.d.\ Gaussian random
variables of variance $1/N_T$, we can write
\begin{equation}
\label{eq:S2g}
|\langle \St \aT , \vS \aT \rangle| = 
|\sum_{l=0}^{N_t-1} \bar{g}_{l-\tau} g_l|, 
\end{equation}
where $g_l \sim {\cal N}(0,1)$. Note that the terms $\bar{g}_{l-\tau} g_l$
in this sum are no longer all jointly independent. But
similar to the proof of Theorem 5.1 in~\cite{PRT07} we observe that
for any $\tau \neq 0$ we can split the index set
${0,\dots,N_t-1}$ into two subsets $\Lambda_\tau^1,\Lambda_\tau^2\subset
\{0,\dots,N_t-1\}$, each of size $N_t/2$, such that the $N_t/2$ variables
$\bar{g}(l-\tau) g(l)$ are jointly independent for $l\in \Lambda^1_\tau$,
and analogous for $\Lambda^2_\tau$. (For convenience we
assume here that $N_t$ is even, but with a negligible modification
the argument also applies for odd $N_t$.)
In other words, each of the sums
$\sum_{l\in \Lambda^r_\tau} \bar{g}(l-\tau) g(l), r=1,2$,
contains only jointly independent terms.
Hence we can apply Lemma~\ref{le:gaussianinner} and obtain
\begin{equation}
\Prob\Big(\big|\sum_{l\in \Lambda^r_\tau}
 \bar{g}(l-\tau) g(l) \big| > t \Big) \le 
2 \exp \Big(-\frac{t^2}{N_t/2 + 2t)}\Big)
\notag
\end{equation}
for all $t>0$. Choosing $t = \frac{3}{2} \sqrt{N_t\log(N_t N_R N_T)}$ gives
\begin{align}
\Prob\Big(\big|\sum_{l\in \Lambda^r_\tau} \bar{g}(l-\tau) g(l) \big| 
  > \frac{3}{2} \sqrt{N_t \log(N_t N_R N_T)} \Big) 
& \le 2\exp \Big(-\frac{\frac{9}{4} N_t \log(N_t N_R N_T)}
{\frac{N_t}{2}+3\sqrt{N_t \log(N_t N_R N_T)}}\Big)\notag \\
& \le 2\exp \Big(-\frac{9 \log(N_t N_R N_T)}
{2+12\sqrt{\frac{\log(N_t N_R N_T)}{N_t}}}\Big).\label{eq:sum5}
\end{align}
Condition~\eqref{eq:coherencecondition} implies that
$12\sqrt{\frac{\log(N_t N_R N_T)}{N_t}} \le \frac{5}{2}$, hence
the estimate in~\eqref{eq:sum5} becomes
\begin{align}
\Prob\Big(\big|\sum_{l\in \Lambda^r_\tau} \bar{g}(l-\tau) g(l) \big| 
  > \frac{3}{2}\sqrt{\log(N_t N_R N_T)} \sqrt{N_t}\Big) 
& \le 2\exp \Big(-\frac{9 \log(N_t N_R N_T)}
{2+\frac{5}{2}} \Big)\notag \\
& = 2\exp \big(-2 \log(N_t N_R N_T) \big) \notag \\
& = 2(N_t N_R N_T)^{-2}.
\label{eq:sum6}
\end{align}
Using equation~\eqref{eq:S2g}, inequality~\eqref{eq:sum6}, and the pigeonhole 
principle, we obtain
\begin{align}
\Prob\Big(|\langle \St \aT , \vS \aT \rangle| > 
3 \sqrt{N_t \log(N_t N_R N_T)} \Big) 
& \le 4 (N_t N_R N_T)^{-2}, \notag
\end{align}
Combining this estimate with~\eqref{innprod} yields
\begin{equation}
\Prob\Big( |\langle \vA_{\tau,\beta}, \vA_{\tau',\beta} \rangle| \ge
3 N_R \sqrt{N_t \log(N_\tau N_R N_T)} \Big) \le 4 (N_t N_R N_T)^{-2},
\notag
\end{equation}
We apply the union bound over the $\frac{N_t}{2} N_T N_R$ different 
possibilities and arrive at
\begin{equation}
\label{eq:case3}
\Prob\Big(\max |\langle \vA_{\tau,\beta}, \vA_{\tau',\beta} \rangle| \le 
3 N_R \sqrt{N_t \log(N_\tau N_R N_T)} \Big) \ge 
1-4(N_t N_R N_T)^{-1},
\end{equation}
where the maximum is taken over all $\tau,\tau',\beta,\beta'$ with
$\tau\neq \tau'$.

An inspection of the bounds~\eqref{eq:case1}, \eqref{eq:case2}, and~\eqref{eq:case3} 
establishes~\eqref{eq:coherencebound}, which is what we wanted to prove.
\end{proof}

The key to proving Theorem~\ref{th:lasso} is to 
combine Lemma~\ref{th:normbound} and Lemma~\ref{th:coherencebound} with
Theorem~\ref{th:CP}.
The latter theorem requires the matrix to have columns of
unit-norm, whereas the columns of our matrix $\vA$ have all different norms
(although the norms concentrate nicely around $\sqrt{N_t N_R N_T}$). 
Thus instead of $\vA \vx = \vy$ we now consider
\begin{equation}
\label{Atilde}
\tilde{\vA} \vxt = \vy, \qquad 
\text{where $\tilde{\vA} := \vA \vD^{-1}$ and $\vxt := \vD \vx$.} 
\end{equation}
Here $\vD$ is the $N_\tau N_\beta \times N_\tau N_\beta$ diagonal matrix 
defined by 
\begin{equation}
\vD_{(\tau,\beta),(\tau,\beta)} = \|\colA\|_2.
\label{diagonalmatrix}
\end{equation}
In the noise-free case we can easily recover $\vx$ from $\vxt$ via 
$\vx = \vD^{-1} \vxt$. In the noisy case we will utilize the fact that for 
proper choices of $\lambda$ the associated lasso solutions of ~\eqref{lassoD} 
and~\eqref{lassoDI}, respectively, have the same support, see also the proof of
Theorem~\ref{th:lasso}.

The following lemma gives a bound for $\mu(\tilde{\vA})$ and
$\|\tilde{\vA}\|_\op$ in terms of the corresponding bounds for $\vA$.

\begin{lemma}
\label{le:tildebounds}
Let $\tilde{\vA} = \vA \vD^{-1}$, where the $\vD$ the diagonal matrix is
defined by~\eqref{diagonalmatrix}. Under the conditions of
Theorem~\ref{th:lasso}, there holds
\begin{equation}
\label{Atilde_estimate}
\Prob \Big( \|\vAt\|_{\op}^2 < 3 (1+\log N_t) \Big) 
\ge 1 - p_1,
\end{equation}
where $p_1 =e^{-N_t\frac{(\sqrt{1/3} - 1)^2}{2}} - N_t^{1 - C \sqrt{N_T}}$,
and
\begin{equation}
\Prob \Big(\mu\big(\tilde{\vA}\big) \le 
  6 \sqrt{\frac{1}{N_t} \log (N_\tau N_R N_T)} \Big) 
\ge 1-p_2,
\label{tildecoherencebound}
\end{equation}
where $p_2=2e^{-\frac{N_t(\sqrt{2}-1)^2}{4}}-2 (N_R N_T)^{-1}-6(N_t N_R N_T)^{-1}$.
\end{lemma}

\begin{proof}
We have
\begin{equation}
\|\tilde{\vA}\|_{\op}^2 \le \frac{\|\vA\|_\op^2}{\max_{\tau,\beta} \|\colA\|_2^2}.
\label{Atbound0}
\end{equation}
Recall that 
\begin{equation}
\colA = \aR \otimes (\St \aT),
\label{Acolumn}
\end{equation}
hence $\|\colA\|_2^2 = \|\aR\|_2^2 \|\St \aT\|_2^2$. Since the entries
$(\St \aT)_k \sim {\cal CN}(0,N_T)$, we have 
$\Exp \|\St \aT\|  = \sqrt{N_t}$, and thus by
Lemma~\ref{le:concentration1}
\begin{equation}
\Prob \Big(\sqrt{N_t} - \|\St \aT\|_2 > t \Big)
\le e^{-\frac{ t^2}{2}},
\label{Atbound1}
\end{equation}
for all $t>0$, hence
\begin{equation}
\Prob \Big(\frac{1}{\|\St \aT\|_2^2} < \frac{1}{(\sqrt{N_t} - t)^2} \Big)
\ge 1 - e^{-\frac{ t^2}{2}},
\label{Atbound2}
\end{equation}
Choosing $t = (1 - \sqrt{1/3})\sqrt{N_t}$ in~\eqref{Atbound2} and forming the 
union bound only over the $N_R N_T$ different possibilities associated with 
$\beta$ (note that $\|\St \aT\|_2 = \|\vS \aT\|_2$ for all $\tau$), gives
\begin{equation}
\Prob \Big( \frac{1}{\underset{\tau,\beta}{\max}\|\colA\|_2^2} < 
\frac{3}{N_t N_R} \Big) \ge 1 - N_R N_T e^{-\frac{N_t(1-\sqrt{1/3})^2}{2}}.
\label{Atbound4}
\end{equation}
The diligent reader may convince herself that the probability
in~\eqref{Atbound4} is indeed close to one under the
condition~\eqref{coherenceproperty2}.
We insert~\eqref{A_estimate} and~\eqref{Atbound4} into~\eqref{Atbound0} and
obtain 
\begin{equation}
\Prob \Big( \|\vAt\|_{\op}^2 < 3N_T (1+\log N_t) \Big) 
\ge 1 - e^{-\frac{N_t(1-\sqrt{1/3})^2}{2}} - N_t^{1 - C \sqrt{N_T}}.
\label{Atbound5}
\end{equation}
which proves~\eqref{Atilde_estimate}.

To establish~\eqref{tildecoherencebound} we first note that
\begin{equation}
\label{coherence3}
\mu(\tilde{\vA}) \le \underset{(\tau,\beta) \neq (\tau',\beta')}{\max} 
\Big\{ \vD^{-1}_{(\tau,\beta),(\tau,\beta)}
|(\vA^{\ast}\vA)_{(\tau,\beta),(\tau',\beta')}| 
\vD^{-1}_{(\tau',\beta'),(\tau',\beta')}\Big\},
\end{equation}
where $\vD^{-1}_{(\tau,\beta),(\tau,\beta)} = \|\colA\|_2^{-1}$.
Using Lemma~\ref{le:concentration1} and~\eqref{Acolumn} we compute
\begin{equation}
\Prob \Big( \|\colA\|_2 > \sqrt{N_t N_R} - \sqrt{N_R} t \Big) 
  \ge 1-e^{-\frac{ t^2}{2}}.
\notag
\end{equation}
Therefore
\begin{equation}
\Prob \Big( \frac{1}{\|\colA\|_2} < 
\frac{1}{\sqrt{N_t N_R} - \sqrt{N_R} t} \Big) \ge 1-e^{-\frac{t^2}{2}},
\notag
\end{equation}
and thus
\begin{equation}
\Prob \Big( |\tilde{\vA}^{\ast} \tilde{\vA})_{(\tau,\beta),(\tau',\beta')}| 
\le \frac{1}{(\sqrt{N_t N_R} - \sqrt{N_R} t)^2} 
|(\vA^{\ast}\vA)_{(\tau,\beta),(\tau',\beta')}| 
\Big) \ge 1-2e^{-\frac{t^2}{2}},
\label{Atbound8}
\notag
\end{equation}
By choosing $t=(1-1/\sqrt{2})\sqrt{N_t}$, we can write~\eqref{Atbound8} as
\begin{equation}
\Prob \Big( |\tilde{\vA}^{\ast} \tilde{\vA})_{(\tau,\beta),(\tau',\beta')}| 
\le \frac{2}{N_t N_R} 
|(\vA^{\ast}\vA)_{(\tau,\beta),(\tau',\beta')}| 
\Big) \ge 1-2e^{-\frac{N_t(\sqrt{2}-1)^2}{4}}.
\label{Atbound9}
\notag
\end{equation}
Finally, plugging~\eqref{Atbound9} into~\eqref{coherence3} and 
using~\eqref{eq:coherencebound} we arrive at
\begin{equation}
\Prob \Big(\mu(\vAt) \le 6 \sqrt{\frac{1}{N_t} \log (N_\tau N_R N_T)} \Big) 
\ge 1-2e^{-\frac{N_t(\sqrt{2}-1)^2}{4}}-2 (N_R N_T)^{-1}-6(N_t N_R N_T)^{-1}.
\notag
\end{equation}
\end{proof}

We are now ready to prove Theorem~\ref{th:lasso}. Among others it
hinges on a (complex version of a) theorem by Cand{\`e}s and Plan~\cite{CP08},
which is stated in Appendix B.

\medskip
\noindent
{\bf Proof of Theorem~\ref{th:lasso}:}
We first point out that the assumptions of Theorem~\ref{th:lasso} imply that 
the conditions of Lemma~\ref{th:normbound} and 
Lemma~\ref{th:coherencebound} are fulfilled. For Lemma~\ref{th:normbound} 
this is obvious. Concerning Lemma~\ref{th:coherencebound}, an easy
calculation shows that the conditions $(\log (N_\tau N_R N_T))^3 \le N_t$ 
and $N_t \ge 128$ indeed yield that $\log (N_t N_R N_T) \le \frac{N_t}{23}$.

Note that the solution $\tilde{\vx}$ of~\eqref{lassoD} and the solution
$\tilde{\vxt}$ of the following lasso problem
\begin{equation}
\underset{\vxt}{\min}\, \frac{1}{2}\|\vA \vD^{-1}\vxt - \vy\|_2^2 + \lambda
\|\vxt\|_1, \qquad \text{with}\,\,\lambda = 2 \sigma\sqrt{2 \log(N_\tau N_R N_T)},
\label{lassoDI}
\end{equation}
satisfy $\supp(\tilde{\vx}) = \supp(\vD^{-1}\tilde{\vxt})$.

We will first establish the claims in Theorem~\ref{th:lasso}
for the system $\vAt \vxt = \vy$ in~\eqref{Atilde} where  $\vAt = \vA
\vD^{-1}$, $\vxt = \vD\vx$ and then switch back to $\vA \vx = \vy$.

We verify first condition~\eqref{amplitudeproperty}. 
Property~\eqref{amplitudeproperty2} and the fact that $\vz = \vD \vx$ imply
that
\begin{equation}
|z_k| \ge \frac{10 \|\vA_{\tau,\beta}\|_2}{\sqrt{N_R N_t}}
\sigma\sqrt{2\log (N_\tau N_\beta)}, \qquad \text{for $(\tau,\beta) \in S$.}
\label{zcond1}
\end{equation}
Using Lemma~\ref{le:concentration1} we get that
\begin{equation}
\Prob\Big(\|\vA_{\tau,\beta}\| \ge \sqrt{N_R N_t} - t\Big) 
    \ge 1 - e^{-\frac{t^2}{2}}.
\label{Anormdev}
\end{equation}
Choosing $t= \frac{2}{10} \sqrt{N_R N_t}$ and combining~\eqref{Anormdev} 
with~\eqref{zcond1} gives
\begin{equation}
|z_k| \ge 8 \sigma \sqrt{2\log (N_\tau N_\beta)}, \qquad \text{for $k \in S$,}
\notag
\end{equation}
with probability at least $1- e^{-\frac{N_R N_t}{25}}$, thus establishing
condition~\eqref{amplitudeproperty}.

Note that $\vAt$ has unit-norm columns as required by Theorem~\ref{th:CP}. 
It remains to verify condition~\eqref{coherenceproperty}.
Using the assumption~\eqref{coherenceproperty2},
and the coherence bound~\eqref{tildecoherencebound} we compute
$$\mu^2(\vAt) \le 36 \frac{1}{N_t}\log(N_\tau N_R N_T) \le
36 \frac{\log(N_\tau N_R N_T)}{\log^3(N_\tau N_R N_T)} = 
\frac{36}{\log^2(N_\tau N_R N_T)},$$
which holds with probability as in~\eqref{tildecoherencebound}, and
thus the coherence property~\eqref{coherenceproperty} is fulfilled.

Furthermore, using~\eqref{Atilde_estimate} we see that
condition~\eqref{lassosparsity1} implies
\begin{equation}
K \le \frac{c_0 N_\tau N_R}{3(1+\log N_t) \log (N_\tau N_R N_T)} \le
\frac{c_0 N_\tau N_R}{\|\vAt\|_{\op}^2 \log (N_\tau N_R N_T)} 
\notag
\end{equation}
with probability as stated in~\eqref{Atilde_estimate}.
Thus assumption~\eqref{lassosparsity} of Theorem~\ref{th:CP} is also
fulfilled (with high probability) and we obtain that
\begin{equation}
\label{support4}
\supp (\tilde{\vxt}) = \supp (\vxt).
\end{equation}
We note that the relation
$\supp(\tilde{\vx}) = \supp(\vx)$ holds with the same probability as
the relation $\supp(\tilde{\vxt}) = \supp(\vxt)$ (see
equation~\eqref{support4}), since $\supp(\vxt) = \supp(\vx)$ and 
multiplication by an invertible diagonal matrix does not change the support 
of a vector. This establishes~\eqref{support2} with the corresponding
probability.

As a consequence of~\eqref{condbound} we have the following error bound 
\begin{equation}
\frac{\|\tilde{\vxt} - \vxt \|_2}{\|\vxt\|_2} 
    \le \frac{3 \sigma \sqrt{N_\tau N_\beta}}{\|\vy\|_2}
\label{error4}
\end{equation}
which holds with probability at least
\begin{equation}
\big(1 - p_1)(1 - p_2\big)(1- e^{-\frac{N_R N_t}{25}})
\big(1 - 2(N_\tau N_\beta)^{-1}(2\pi \log (N_\tau N_\beta) + K(N_\tau N_\beta)^{-1})
- {\cal O}((N_\tau N_\beta)^{-2 \log 2})\big),
\notag
\end{equation}
where the probabilities $p_1, p_2$ are as in~Lemma~\ref{le:tildebounds}.
Using the fact that $\tilde{\vz} = \vD \tilde{\vx}$, we compute
$$
\frac{1}{\kappa(\vD)} \frac{\|\tilde{\vx} - \vx \|_2}{\|\vx\|_2}
\le \frac{\|\vD(\tilde{\vx} - \vx) \|_2}{\|\vD\vx\|_2}  =
\frac{\|\tilde{\vxt} - \vxt \|_2}{\|\vxt\|_2},
$$
or, equivalently,
\begin{equation}
 \frac{\|\tilde{\vx} - \vx \|_2}{\|\vx\|_2} \le
\kappa(\vD) \frac{\|\tilde{\vxt} - \vxt \|_2}{\|\vxt\|_2}.
\label{z2x}
\end{equation}
Proceeding along the lines of~\eqref{Atbound1}-\eqref{Atbound4}, we estimate
\begin{equation}
\Prob \big(\kappa (\vD) \le 2 \big) 
\ge 1 - N_R N_T e^{-\frac{N_t(1-\sqrt{1/3})^2}{2}}.
\label{condD}
\end{equation}
The bound~\eqref{error2} follows now from combining~\eqref{error4}
with~\eqref{z2x} and~\eqref{condD}.
\QED


\section{Recovery of targets in the Doppler case}
\label{s:doppler}

In this section we analyze the case of moving targets/antennas, as 
described in~\ref{ss:doppler}. As in the stationary setting, we
assume that $s_i(t)$ is a periodic, continuous-time white Gaussian
noise signal of period-duration $T$ seconds and bandwidth $B$. The transmit
waveforms are normalized so that the total transmit power is fixed,
independent of the number of transmit antennas. Thus, we assume that
the entries of $s_i(t)$ have variance $\frac{1}{N_T}$.

\begin{theorem}
\label{th:doppler}
Consider $\vy = \vA \vx +\vw$, where $\vA$ is as defined in
Subsection~\ref{ss:doppler} and $\vw_i \in {\cal CN}(0,\sigma^2)$.
Choose the discretization stepsizes to be $\Delta_\beta = \frac{2}{N_R N_T}$, 
$\Delta_\tau = \frac{1}{2B}$ and $\Delta_f = \frac{1}{T}$. 
Let $d_T = 1/2, d_R = N_T/2$ or $d_T = N_R/2, d_R = 1/2$, and suppose that 
\begin{equation*}
N_t \ge 128, \qquad \max\{N_\tau,N_f,\sqrt{N_\tau,N_f}\} \ge \sqrt{N_\beta}, 
\qquad \text{and} \qquad \big(\log (N_\tau N_\beta) \big)^3 \le N_t.
\end{equation*}
If $\vx$ is drawn from the generic $K$-sparse target model with
\begin{equation*}
K \le \Kmax := \frac{c_0 N_\tau N_f N_R}{6 \log (N_\tau N_f N_\beta)}
\end{equation*}
for some constant $c_0>0$, and if
\begin{equation*}
\underset{k\in I}{\min}\, |\vx_k| > \frac{10 \sigma}{\sqrt{N_R N_t}} \sqrt{2 \log N_\tau N_f N_\beta},
\end{equation*}
then the solution $\tilde{\vx}$ of the debiased lasso computed with
$\lambda = 2\sigma \sqrt{2 \log(N_\tau N_f N_\beta)}$
obeys
\begin{equation*}
\supp (\tilde{\vx}) = \supp (\vx),
\end{equation*}
with probability at least 
\begin{equation*}
\notag
(1 - p_1)(1 - p_2)(1-p_3)(1-p_4),
\end{equation*}
and
\begin{equation*}
\frac{\|\tilde{\vx} - \vx \|_2}{\|\vx\|_2} 
    \le \frac{ \sigma \sqrt{12 N_t N_R}}{\|\vy\|_2}
\end{equation*}
with probability at least 
\begin{equation*}
\notag
(1 - p_1)(1 - p_2)(1-p_3)(1-p_4)(1 - p_5),
\end{equation*}
where
$$p_1=e^{-\frac{(1-\sqrt{1/3})^2 N_t}{2}}+N_T e^{-(\sqrt{3/2}-\sqrt{2}) N_t},$$
$$ 
p_2= 2 (N_R N_T)^{-1} + 2 (N_\tau N_R N_T)^{-1} + 2 (N_f N_R N_T)^{-1}
+ 6 (N_\tau N_f N_R N_T)^{-1} + 2e^{-\frac{N_t(\sqrt{2}-1)^2}{4}},$$
$$p_3 = N_R N_T e^{-\frac{(1-\sqrt{1/3})^2 N_t}{2}}, \qquad
p_4 = e^{-\frac{N_R N_t}{25}},$$
and
$$
p_5 = 2(N_\tau N_\beta)^{-1}(2\pi \log (N_\tau N_\beta) + S(N_\tau N_\beta)^{-1})
+ {\cal O}((N_\tau N_\beta)^{-2 \log 2}).$$

\end{theorem}

\begin{proof}
The proof is very similar to that of Theorem~\ref{th:lasso}. Below 
we will establish the analogs of the key steps, Lemma~\ref{th:normbound},
Lemma~\ref{th:coherencebound}, and Lemma~\ref{le:tildebounds},
and leave the rest to the reader.
\end{proof}

\begin{lemma}
\label{th:normboundoppler}
Let $\vA$ be as defined in Theorem~\ref{th:doppler}. Then
\begin{equation}
\label{DA_estimate}
\Prob \Big(\|\vA\|^2_{\text{op}} \le 2 N_t N_f N_R N_T \Big) \ge
1 - N_T e^{-N_t(\frac{3}{2}-\sqrt{2})}.
\end{equation}
\end{lemma}

\begin{proof}
We proceed as in the proof of Lemma~\ref{th:normbound}.
There holds $\|\vA\|_\op^2 = \|\vA \vA^{\ast}\|_\op$.
It is convenient to consider $\vA \vA^{\ast}$ as block matrix
$$
\begin{bmatrix}
\vB_{1,1}            & \vB_{1,2}   & \dots      & \vB_{1,N_R} \\
\vdots             & \ddots    &            & \vdots    \\
\vB_{N_R,1}^{\ast}   &           &            & \vB_{N_R,N_R}
\end{bmatrix},
$$
where the blocks $\{\vB_{i,i'}\}_{i,i'=1}^{N_R}$ are matrices of size 
$N_t \times N_t$. We claim that $\vA \vA^{\ast}$ is a block-Toeplitz matrix 
(i.e., $\vB_{i,i'} = \vB_{i+1,i'+1}, i=1,\dots, N_R-1$) and the individual 
blocks $\vB_{i,i'}$ are circulant matrices. To see this, recall the structure 
of $\vA$ and consider the entry $\vB_{[i,l;i',l']}$, $i,i'=1,\dots,N_R;
l,l'=1,\dots,N_t$:
\begin{gather}
\vB_{[i,l;i',l']}  =  (\vA \vA^{\ast})_{[i,l;i',l']} = 
\sum_{\beta} \sum_{\tau} \sum_{f}^{}\vA_{[i,l;\tau,f,\beta]} 
\vA_{[i',l';\tau,f,\beta]}\notag \\
 =  \sum_{\beta} e^{j2\pi d_R (i-i') \beta}
\sum_{k=1}^{N_T} \sum_{k'=1}^{N_T} e^{j2\pi d_T (k-k')\beta}
G_{k,k'}(l,l') \sum_{m=1}^{N_f} e^{j2\pi (l-l')\Delta_t m\Delta_f } \notag\\
 =  \sum_{n=0}^{N_R N_T-1} e^{j2\pi (i-i')\frac{nN_T}{N_R N_T}}
\sum_{k=1}^{N_T} \sum_{k'=1}^{N_T} e^{j2\pi (k-k')\frac{n}{N_R N_T}}
G_{k,k'}(l,l') N_f \delta_{l-l'} \label{eq:expsum}\\
= N_T N_R N_f \sum_{k=1}^{N_T} \|\vs_{k}\|^2 \delta_{i-i'}\delta_{l-l'} 
\label{matrixstructuredoppler}
\end{gather}
where we have used in~\eqref{eq:expsum} that $N_f = \frac{2B}{\Delta_f} =
2BT$, whence $\sum_{m=1}^{N_f} e^{j2\pi (l-l')m\Delta_t \Delta_f} = N_f \delta_{l-l'}$.
Thus
\begin{equation}
\label{eq:A2I}
\vA\vA^{\ast} =  (N_T N_R N_f \sum_{k=1}^{N_T} \|\vs_{k}\|^2) \, \vI,
\end{equation}
i.e., $\vA \vA^{\ast}$ is just a scaled identity matrix.
Since $\vs_k$ is a Gaussian random vector with $\vs_k(j) \sim {\cal CN}(0,1)$,
Lemma~\ref{le:concentration1} yields
\begin{equation}
\Prob \Big( \|\vs_k\|_2^2 - (\Exp \|\vs_k\|_2)^2  \ge t(t+2\Exp \|\vs_k\|_2)
\Big) \le e^{- t^2/2},
\label{gaussiannorms}
\end{equation}
where we note that $\Exp \|\vs_k\|_2 = \sqrt{\frac{N_t}{N_T}}$.
We choose $t = (\sqrt{2}-1)\sqrt{N_t}$, and obtain, after forming the
union bound over $k=1,\dots,N_t-1$, 
\begin{equation}
\Prob \Big( \sum_{k=1}^{N_T}\|\vs_k\|_2^2  )^2  \ge 2N_t \Big) 
\le N_T e^{-N_t(\frac{3}{2}-\sqrt{2})}.
\label{gaussiannorms1}
\end{equation}
The bound~\eqref{DA_estimate} now follows from~\eqref{eq:A2I}.
\end{proof}

Next we establish a coherence bound for $\vA$.
\begin{lemma}
\label{th:coherencedoppler}
Let $\vA$ be as defined in the Doppler case. Assume that 
\begin{equation}
N \ge \sqrt{N_{\beta}}
\label{eq:coherencecondition1}
\log(N N_\beta) < \frac{N_t}{30},
\end{equation}
where $N: = \max\{N_\tau,N_f,\sqrt{N_\tau N_f}\}$. Then
\begin{equation}
\underset{(\tau,\dopp,\beta)\neq (\tau',\dopp',\beta')}{\max} 
\big|\langle  \vA_{\tau,\dopp,\beta},\vA_{\tau',\dopp',\beta'} \rangle \big|
\le 3 N_R \sqrt{N_t \log(N_\tau N_f N_\beta)} 
\notag
\end{equation}
with probability at least 
$1 - 2 (N_R N_T)^{-1} - 2 (N_\tau N_R N_T)^{-1} - 2 (N_f N_R N_T)^{-1}
- 6 (N_\tau N_f N_R N_T)^{-1}$.
\end{lemma}

\begin{proof}

We have that $\vA_{\tau,\dopp,\beta} = \aR \otimes (\Std \aT)$.
A standard calculation shows that
\begin{equation}
\label{eq:shifttau1}
|\langle \Std \aT , \Stdp \aTp \rangle| =
|\langle \vS_{\tau - \tau',f-f'} \aT , \aTp \rangle|
\end{equation}
for $\tau,\tau'=0,\dots,N_\tau-1, f,f'=0,\dots,N_f-1$,
thus we only need to consider $|\langle \Std \aT , \vS \aTp \rangle|$.
As in the proof of Lemma~\ref{th:coherencebound} we distinguish
several cases. 

\noindent
{\bf Case (a) $\beta\neq \beta', \tau =0, f=0$:}
In this case we are concerned with $|\langle \vS \aT , \vS \aTp \rangle|$, 
which is the same as Case (i) of Lemma~\ref{th:coherencebound}, except
that in the present case we have a bit more flexibility in choosing $t$
in the analogous version of~\eqref{innprod4}. Here we can choose 
$t= 3\sqrt{\frac{1}{N_t} \log (N N_R N_T)}$, where 
$N = \max\{N_\tau,N_f,\sqrt{N_\tau N_f}\}$. Proceeding then as in 
the proof of Case~(i) of Lemma~\ref{th:coherencebound} we obtain
\begin{equation}
\Prob\Big( |\langle \vA_{\tau,f,\beta},\vA_{\tau,f,\beta'}\rangle| \le 
3 N_R \sqrt{N_t \log(N_\tau N_R N_T)} \Big) \ge 
1 - 2 (N_R N_T)^{-1}. 
\label{eq:casea}
\end{equation}

\medskip
\noindent
{\bf Case (b) $\beta\neq \beta', \tau \neq 0, f = 0$:}
This is exactly the same as Case (ii) of Lemma~\ref{th:coherencebound}.
We obtain
\begin{equation}
\Prob\Big( |\langle \vA_{\tau,f,\beta},\vA_{\tau',f,\beta'}\rangle| \le 
3 N_R \sqrt{N_t \log(N_\tau N_R N_T)} \Big) \ge 
1 - 2 (N_\tau N_R N_T)^{-1}. 
\label{eq:caseb}
\end{equation}

\medskip
\noindent
{\bf Case (c) $\beta\neq \beta', \tau  = 0, f \neq 0$:}
It is well known that $(\vT_\tau \vx)^{\wedge} = \vM_{-\tau}\hat{\vx}$.
Hence, by Parseval's theorem, $\langle \vT_\tau \vx, \vy \rangle=
\langle \vM_{-\tau} \hat{\vx},\hat{\vy} \rangle$.
Since the normal distribution is invariant under Fourier transform, 
this case is therefore already covered by Case (b), and we leave
the details to the reader. We get
\begin{equation}
\Prob\Big( |\langle \vA_{\tau,f,\beta},\vA_{\tau,f',\beta'}\rangle| \le 
3 N_R \sqrt{N_t \log(N_f N_R N_T)} \Big) \ge 
1 - 2 (N_f N_R N_T)^{-1}. 
\label{eq:casec}
\end{equation}


\medskip
\noindent
{\bf Case (d) $\beta\neq \beta', \tau  \neq 0, f \neq 0$:}
This is similar to Case (ii) of Lemma~\ref{th:coherencebound}. The only
difference is that we have $N_t N_f N_R N_T$ different possibilities
to consider when forming the union bound (the additional factor $N_f$
is of course due to frequency shifts associated with the Doppler effect).
Thus in this case the bound reads 
\begin{equation}
\Prob\Big( |\langle \vA_{\tau,f,\beta}, \vA_{\tau',f',\beta'} \rangle| \le 
3 N_R \sqrt{N_t \log(N_\tau N_f N_R N_T)} \Big) \ge 
1 - 2 (N_\tau N_f  N_R N_T)^{-1}. 
\label{eq:case4d}
\end{equation}

\medskip
\noindent
{\bf Case (e) $\beta  = \beta'$:} We need to bound
$|\langle  \vT_\tau \vM_f \vS\aT , \vS \aT \rangle|$, where we recall
that $\vS \aT$ is a Gaussian random vector with variance $N_T$. 
(We note that a related case is covered by Theorem 5.1 in~\cite{PRT07}, 
which considers $\langle T_\tau M_f h,h \rangle$, where $h$ is a 
Steinhaus sequence.) This case is essentially taken care off by Case (iii) 
of Lemma~\ref{th:coherencebound}, by noting that a Gaussian random vector
of variance $\sigma$ remains Gaussian (with the same $\sigma$) when
pointwise multiplied by a fixed vector with entries from the
torus. The only difference is that, as in Case (d) above, 
we have $N_t N_f N_R N_T$ different possibilities to consider when 
forming the union bound. Hence, the bound in this case becomes
\begin{equation}
\label{eq:casee}
\Prob\Big(\max |\langle \vA_{\tau,f,\beta}, \vA_{\tau',f',\beta} \rangle| \le 
3 N_R \sqrt{N_t \log(N_\tau N_f N_R N_T)} \Big) \ge 
1-4(N_t N_f N_R N_T)^{-1}.
\end{equation}
\end{proof}

\begin{lemma}
\label{le:tildeboundsdoppler}
Let $\tilde{\vA} = \vA \vD^{-1}$, where the entries of
the $N_\tau N_f N_\beta \times N_\tau N_f N_\beta$ diagonal matrix are
given by $\vD_{(\tau,f,\beta),(\tau,f,\beta)} = \|\colA\|_2$.
Under the conditions of Theorem~\ref{th:lasso} there holds
\begin{equation}
\label{Atilde_estimatedoppler}
\Prob \Big( \|\vAt\|_{\op}^2 < 6 N_T  \Big) 
\ge 1 - p_1,
\end{equation}
where 
$$p_1=e^{-\frac{(1-\sqrt{1/3})^2 N_t}{2}} + N_T e^{-(\sqrt{3/2}-\sqrt{2})
N_t},$$
and
\begin{equation}
\Prob \Big(\mu\big(\tilde{\vA}\big) \le 
  6 \sqrt{\frac{1}{N_t} \log (N_\tau N_f N_R N_T)} \Big) 
\ge 1-p_2,
\label{tildecoherencebounddoppler}
\end{equation}
where 
$$p_2=
 2 (N_R N_T)^{-1} + 2 (N_\tau N_R N_T)^{-1} + 2 (N_f N_R N_T)^{-1}
+ 6 (N_\tau N_f N_R N_T)^{-1} + 2e^{-\frac{N_t(\sqrt{2}-1)^2}{4}},$$
\end{lemma}

\begin{proof}
Since the proof of this lemma follows closely that of
Lemma~\ref{le:tildebounds}, we omit it.
\end{proof}

\section{Numerical Experiments}
\label{s:numerics}

Next we illustrate the performance of the compressive MIMO radar developed
in previous sections. We consider a Doppler-free scenario.
The following parameters are used in this example: 
$N_T = 8$ transmit antennas, $N_R = 8$ receive antennas, 
$N_t = 64$ samples, $N_\tau = N_t$ range values.

At each experiment $K$ scatterers of unit amplitude are placed randomly
on the range/azimuth grid, i.e the vector $\vx$ has $K$ unit entries at
random locations along the vector. White Gaussian noise is added to the
composite data vector $\vA \vx$ with variance $\sigma^2$ determined to
as to produce the specified output signal-to-noise ratio (see also
item (iv) of the Remark after Theorem~\ref{th:lasso}). The lasso 
solution $\hat{\vx}$ is calculated with $\lambda$ as specified in
Theorem~\ref{th:lasso}. The numerical algorithm to
solve~\eqref{lassoD} was implemented in Matlab using TFOCS~\cite{BCG10}. 
The experiment is repeated $100$ times
using independent noise realizations.

The probabilities of detection $P_d$ and false alarm $P_{fa}$ are
computed as follows. The values of the estimated vector $\hat{\vx}$
corresponding to the true scatterer locations are compared to a
threshold. Detection is declared whenever a value exceeds the threshold.
The probability of detection is defined as the number of detections
divided by the total number of scatterers $K$. Next the values of the
estimated vector $\hat{\vx}$ corresponding to locations not containing
scatterers are compared to a threshold. A false alarm is declared
whenever one of these values exceeds the threshold. The probability of
false alarm is defined as the number of false alarms divided by the
total number of scatterers $K$. The probabilities of detection and false
alarm are averaged over the 100 repetitions of the experiment.

The probabilities are re-computed for a range of values of the threshold
to produce the so-called Receiver Operating Characteristics (ROC)
\cite{helstrom,vantrees,scharf} - the graph of $P_d$ vs. $P_{fa}$.   
As the threshold decreases, the probability of detection increases and
so does the probability of false alarm. In practice the threshold is
usually adjusted to as to achieve a specified probability of false
alarm.

Figures \ref{fig1}, \ref{fig2}, \ref{fig3} and \ref{fig4} depict the ROC
for different values of the output signal to noise ratio. We note that
the probability of detection increases as the SNR increases and
decreases as $K$, the number of scatterers increases. 


\begin{figure}[ht!]
\centering
\includegraphics[width=5in,height=3.3in,clip,keepaspectratio]{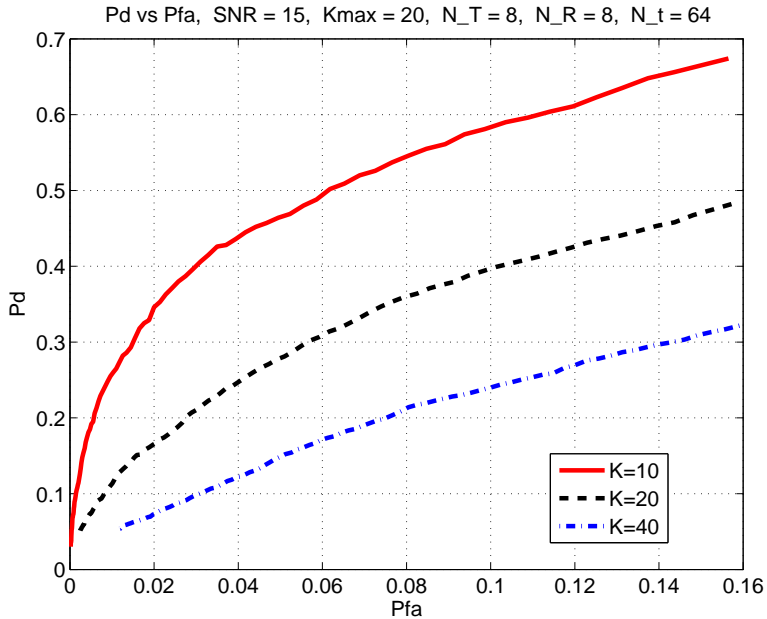}
\caption{Probability of detection vs. probability of false alarm for
SNR = 15 dB, and three values of $K$: $K_{\max}/2, K_{\max}, 2K_{\max}$.}
\label{fig1}
\end{figure}

\begin{figure}[ht!]
\centering
\includegraphics[width=5in,height=3.3in,clip,keepaspectratio]{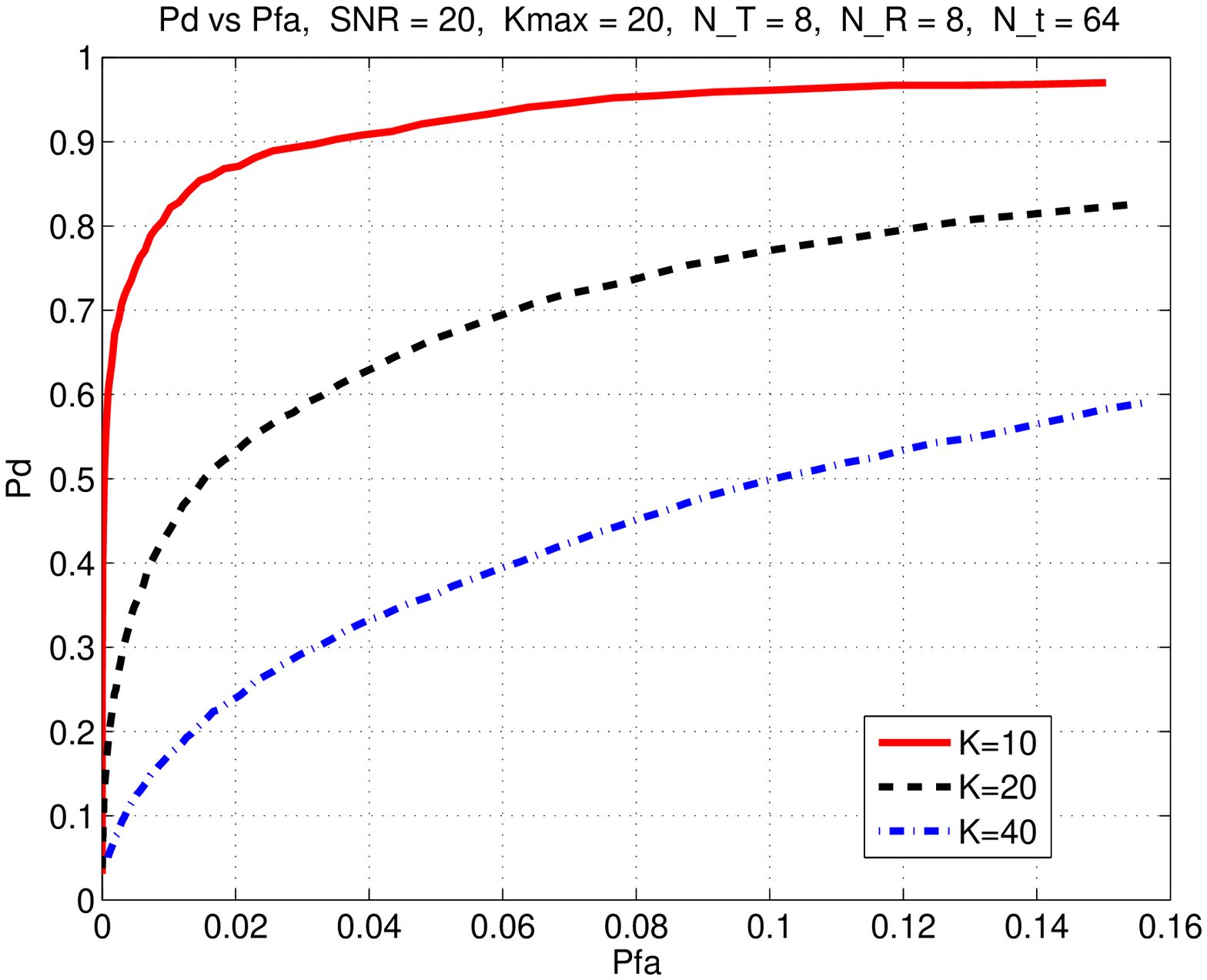}
\caption{Probability of detection vs. probability of false alarm for
SNR = 20 dB, and three values of $K$: $K_{\max}/2, K_{\max}, 2K_{\max}$.}
\label{fig2}
\end{figure}

\begin{figure}[ht!]
\centering
\includegraphics[width=5in,height=3.3in,clip,keepaspectratio]{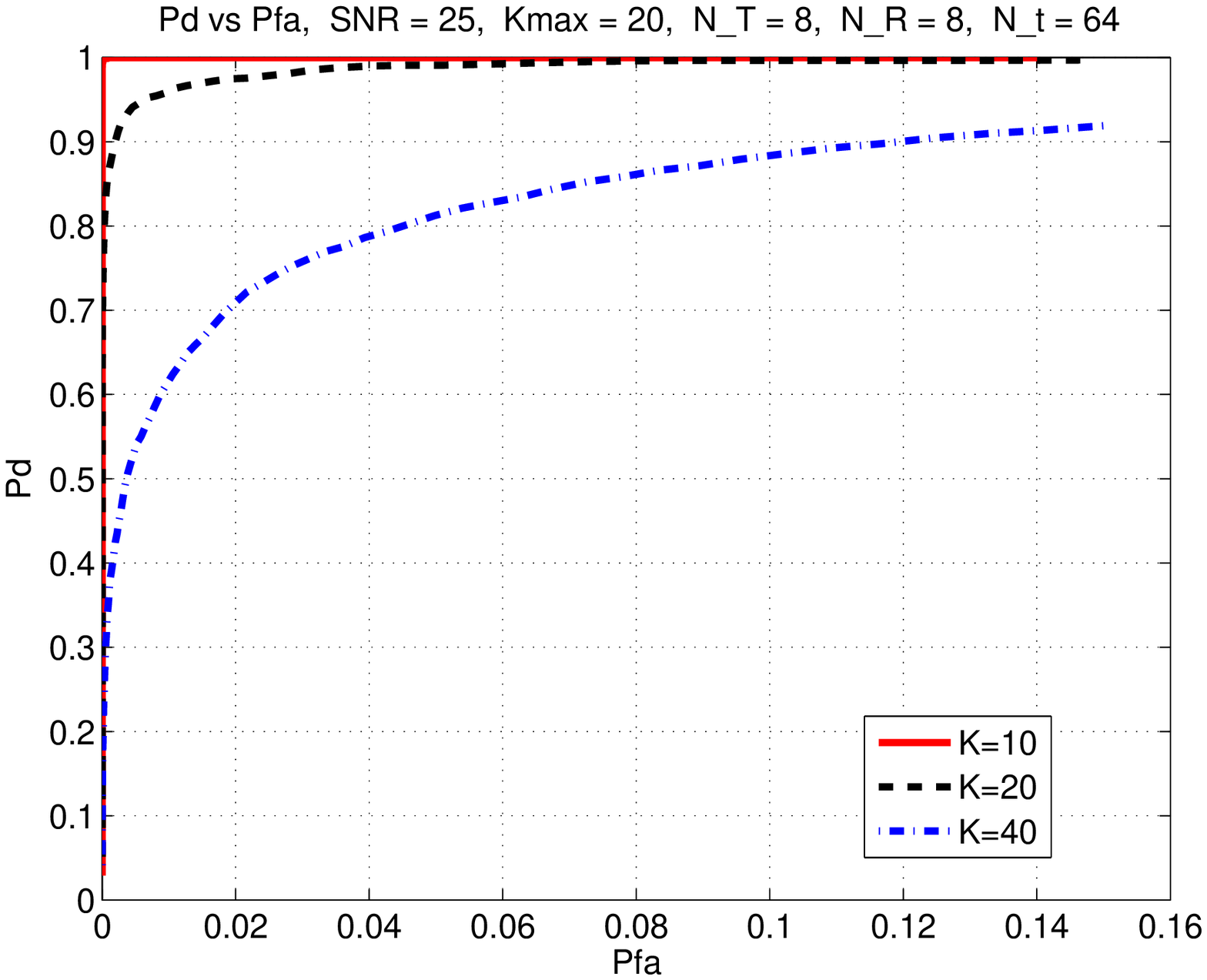}
\caption{Probability of detection vs. probability of false alarm for
SNR = 25 dB, and three values of $K$: $K_{\max}/2, K_{\max}, 2K_{\max}$.}
\label{fig3}
\end{figure}

\begin{figure}[ht!]
\centering
\includegraphics[width=5in,height=3.3in,clip,keepaspectratio]{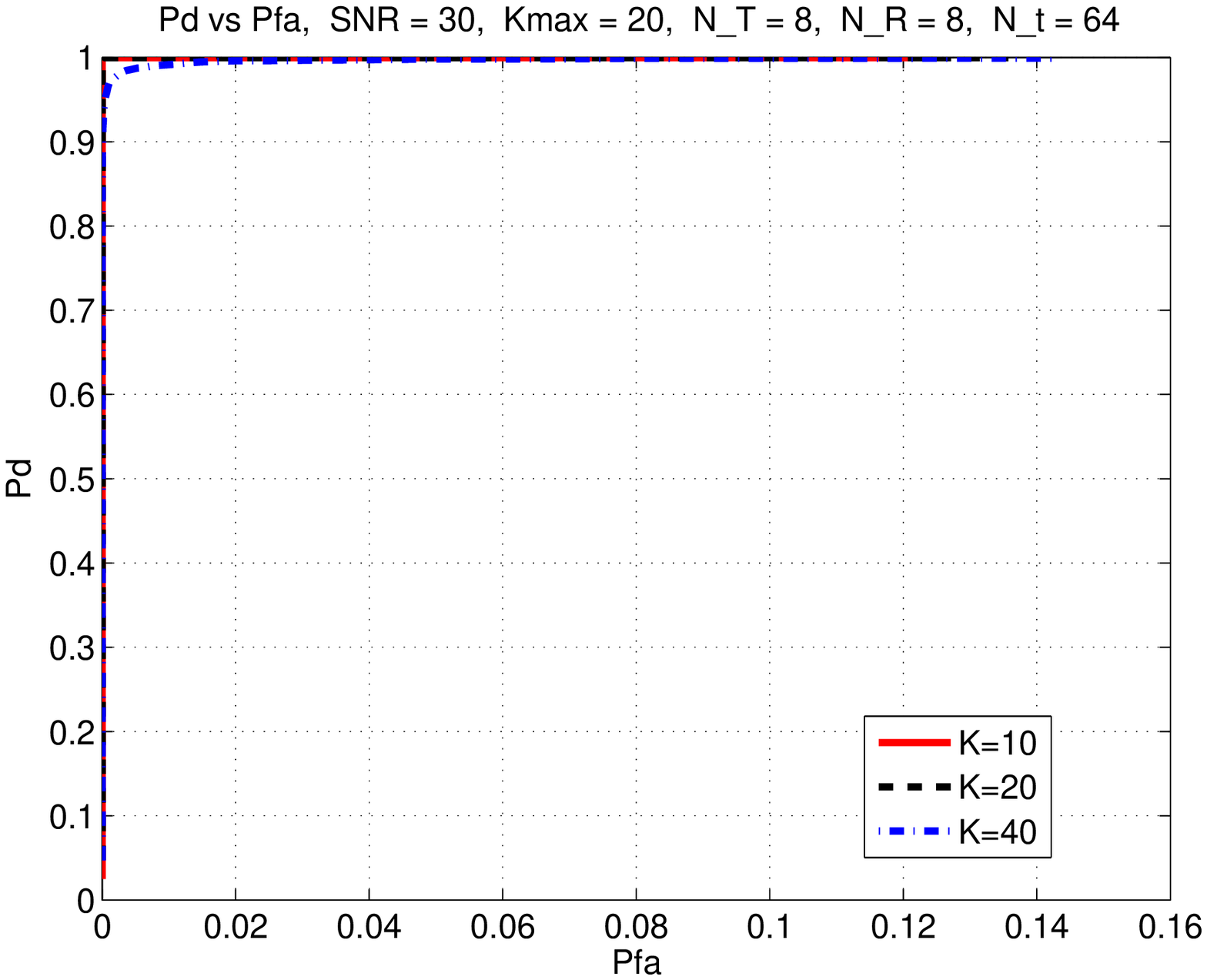}
\caption{Probability of detection vs. probability of false alarm for
SNR = 30 dB, and three values of $K$: $K_{\max}/2, K_{\max}, 2K_{\max}$.}
\label{fig4}
\end{figure}

\section{Conclusion}
\label{s:conclusion}


Techniques from compressive sensing and sparse approximation make it 
possible to exploit the sparseness of radar scenes to potentially improve 
system performance of MIMO radar. In this paper we have derived a mathematical
framework that yields explicit conditions
for the radar waveforms and the transmit and receive arrays so that
the radar sensing matrix has small coherence and robust sparse recovery
in the presence of noise becomes possible.
Our approach relies on a deterministic (and very specific) positioning
of transmit and receive antennas and random waveforms. It seems plausible
that results similar to the ones derived in this paper can be established
for the case where the antenna locations are chosen at random and the
transmission signals are deterministic. This would be of interest, since 
one could then potentially take advantage of specific properties of recently 
designed deterministic radar waveforms such as in~\cite{BD10,PCM08}.

\section*{Appendix A} \label{s:appendixA}

In this appendix we collect some auxiliary results.

\begin{lemma}\cite[Proposition 34]{Ver10}
\label{le:concentration1}
Let $\vx \in \CC^n$ be a vector with $x_k \sim {\cal CN}(0,\sigma^2)$,
then for every $t >0$ one has
\begin{equation}
\Prob \Big( \|\vx\|_2 - \Exp \|\vx\|_2 > t \Big)
\le e^{-\frac{t^2}{2\sigma^2}}.
\label{concentration2}
\end{equation}
\end{lemma}

The following lemma, which relates moments and tails, can be found 
e.g.\ in~\cite[Proposition 6.5]{Rau10}.
\begin{lemma}
\label{le:moments}
Suppose $Z$ is a random variable satisfying
\begin{equation}
(\Exp |Z|^p)^{1/p} \le \alpha \beta^{1/p} p^{1/\gamma} \qquad
\text{for all $p \ge p_0$}
\label{eq:moments1}
\notag
\end{equation}
for some constants $\alpha, \beta, \gamma, p_0 > 0$. Then
\begin{equation}
\Prob (|Z| \ge e^{1/\gamma} \alpha u ) \le \beta e^{-u^{\gamma}/\gamma}
\notag
\end{equation}
for all $u \ge p_0^{1/\gamma}$.
\end{lemma}

The following lemma is a rescaled version of Lemma 3.1 in~\cite{RSV08}.
\begin{lemma}
\label{le:concentration3}
Let $\vA \in \CC^{n \times m}$ be a Gaussian random matrix with
$A_{i,j} \sim {\cal CN}(0,\sigma^2)$. Then for all $\vx, \vy \in\CC^m$
with $\|\vx\|_2 = \|\vy\|_2=\sqrt{m}$ and all $t>0$
\begin{equation}
\Prob \Big\{ |\frac{1}{n \sigma^2}\langle \vA\vx, \vA\vy \rangle - \langle \vx,\vy \rangle| > tm \Big\} \le
2 \exp \Big(-n\frac{t^2}{C_1+C_2t}\Big),
\notag
\end{equation}
with $C_1 = \frac{4e}{\sqrt{6\pi}}$ and $C_2 = \sqrt{8}e$.
\end{lemma}

The next lemma is a slight 
generalization of a result by Hanson and Wright on tail bounds for 
quadratic forms~\cite{HW71}.
\begin{lemma}
\label{le:quadform}
Let $\vM = \{m_{ij}\}_{i,j=1}^n$ be a normal matrix and let $X_i,
i=0,\dots,n-1$ be independent, ${\cal CN}(0,1)$-distributed random
variables. Denote
\begin{equation}
S_n = \sum_{i,j=0}^{n-1} m_{ij} X_i \bar{X}_j .
\notag
\end{equation}
Then for all $t>0$
\begin{equation}
\Prob\Big(S_n \ge t + \Exp S_n \Big) \le \exp \big( - C 
\min \{\frac{t}{\sigma \|\vM\|_{\op}}, \frac{t^2}{\sigma^2\|\vM\|_F^2}\} \big),
\notag
\end{equation}
where $C$ is a numerical constant independent of $\vM$ and $n$.
\end{lemma}

\begin{proof}
The proof 
follows essentially the same steps as the proof of the main theorem
in~\cite{HW71}, which considers the case where $\vM$ is hermitian and
the $x_i$ are real-valued. Extending the $x_i$ to the complex case is
trivial, thus the only modification that needs to be addressed is the 
extension of $\vM$ from the hermitian to the normal case. But Lemma~5 
in~\cite{HW71} holds for normal matrices as well, therefore the lemma follows.
\end{proof}

For convenience we state the following version of Bernstein's inequality,
which will be used in the proof of Lemma~\ref{le:gaussianinner}.
\begin{theorem}[See e.g.~\cite{VW96}]
\label{th:bernstein}
Let $X_1,\dots,X_n$ be independent random variables with zero mean such
that
\begin{equation*}
\Exp |X_i|^p \le \frac{1}{2} p! K^{p-2} v_i, 
    \qquad \text{for all $i=1,\dots,n; p\in\NN, p\ge 2$},
\end{equation*}
for some constants $K>0$ and $v_i > 0, i=1,\dots,n$. Then, for all $t>0$
\begin{equation}
\Prob\Big( \big|\sum_{i=1}^{n} X_i | \ge t \big) \le 
   2 \exp\Big( - \frac{t^2}{2v + Kt}\Big),
\label{eq:bernstein2}
\end{equation}
where $v:=\sum_{i=1}^{n} v_i$.
\end{theorem}

We also need the following deviation inequality for unbounded
random variables. It is a complex-valued and slightly sharpened 
version of Lemma~6 in \cite{HBR10}, the better constant will be useful
when we apply Lemma~\ref{le:gaussianinner} in the proof of
Lemma~\ref{th:coherencebound}.
\begin{lemma}
\label{le:gaussianinner}
Let $X_i$ and $Y_i$, $i=1,\dots,n$, be sequences of i.i.d.\ complex
Gaussian random variables with variance $\sigma$. Then,
\begin{equation}
\Prob\Big( \big|\sum_{i=1}^{n} \bar{X}_i Y_i \big| > t \Big) \le
2 \exp \big(-\frac{t^2}{\sigma^2 (n \sigma^2 + 2t)}\big).
\label{eq:gaussianinner}
\end{equation}
\end{lemma}

\begin{proof}
In order to apply Bernstein's inequality, we need
to compute the moments $\Exp |X_i Y_i|^p$.
Since $X_i$ and $Y_i$ are independent, there holds
\begin{equation*}
\Exp (|X_i Y_i|^p ) = \Exp (|X_i|^p ) \Exp (|Y_i|^p ) =(\Exp (|X_i|^p ))^2.
\end{equation*}
The moments of $X_i$ are well-known:
\begin{equation*}
\Exp |X_i|^{2p} = p! \, \sigma^{2p},
\end{equation*}
hence
\begin{equation*}
(\Exp |X_i|^{2p})^2 = (2p!)^2 (\sigma^{2p})^2 \le 
\frac{1}{4} (2p)!  (\sigma^{2})^{2p} \le
\frac{1}{2} (2p)!  (\sigma^{2})^{2p-2} \frac{(\sigma^{2})^2}{2}.
\end{equation*}
We apply Bernstein's inequality~\eqref{eq:bernstein2} with $K= \sigma^2$ and 
$v_i = \frac{(\sigma^{2})^2}{2}, i=1,\dots,n$ and
obtain~\eqref{eq:gaussianinner}.
\end{proof}

\section*{Appendix B} \label{s:appendixB}

We consider a general linear system of equations $\vPsi \vx = \vy$,
where $\vPsi \in \CC^{n \times m}$, $\vx \in \CC^m$
and $n \le m$. We introduce the following generic $K$-sparse model:
\begin{itemize}
\item The support $I \subset \{1,\dots,m\}$ of the $K$ nonzero
coefficients of $\vx$ is selected uniformly at random.
\item The non-zero entries of $\sgn(\vx)$ form a Steinhaus sequence, i.e.,
$\sgn(\vx_k):=\vx_k/|\vx_k|, k\in I,$ is a complex random variable that is 
uniformly distributed on the unit circle.
\end{itemize}

The following theorem is a slightly extended version of 
Theorem~1.3 in~\cite{CP08}. 

\begin{theorem}
\label{th:CP}
Given $\vy = \vPsi \vx + \vw$, where $\vmat$ has all unit-$\ell_2$-norm 
columns, $\vx$ is drawn from the generic $K$-sparse model and 
$\vw_i \sim {\cal CN}(0,\sigma^2)$. Assume that
\begin{equation}
\mu(\vPsi) \le \frac{C_0}{\log m},
\label{coherenceproperty}
\end{equation}
where $C_0 >0$ is a constant independent of $n,m$.
Furthermore, suppose 
\begin{equation}
K \le \frac{c_0 m}{\|\vPsi \|_{\op}^2 \log m}
\label{lassosparsity}
\end{equation}
for some constant $c_0 > 0$ and that
\begin{equation}
\underset{k\in I}{\min}\, |\vx_k| > 8 \sigma \sqrt{2 \log m}.
\label{amplitudeproperty}
\end{equation}
Then the solution $\hat{\vx}$ to the debiased lasso computed with
$\lambda = 2 \sigma \sqrt{2 \log m}$
obeys
\begin{equation}
\label{supportbound}
\supp (\hat{\vx}) = \supp (\vx),
\end{equation}
and
\begin{equation}
\frac{\|\hat{\vx} - \vx \|_2}{\|\vx\|_2} 
    \le \frac{\sigma \sqrt{3 n}}{\|\vy\|_2}
\label{condbound}
\end{equation}
with probability at least
\begin{equation}
1 - 2m^{-1}(2\pi \log m + Km^{-1}) - {\cal O}(m^{-2 \log 2}).
\end{equation}
\end{theorem}

\begin{proof}
The paper~\cite{CP08} treats only the real-values case. However it 
is not difficult to see that the results by Cand{\`e}s and Plan can be 
extended to the complex setting if their definition of the sign-function
is replaced by~\eqref{sgn} and consequently their generic sparse model
is replaced by the generic sparsity model introduced in the beginning of
this appendix. The proofs of the theorems in~\cite{CP08} can
then be easily adapted to the complex case via some straightforward
modifications, such as replacing in many steps $\langle \cdot,\cdot
\rangle$ by its real part, $\Real \langle \cdot,\cdot \rangle$ and
replacing certain scalar quantities by its conjugate analogs.
To give a concrete example of such a modification, consider
(in the notation of~\cite{CP08}) the inequality right before 
eq.(3.10) in~\cite{CP08}, 
$$|\hat{\beta}_i| = |\beta_i + h_i| \ge |\beta_i| + \sgn(\beta_i)h_i.$$
This inequality needs to be replaced by its complex counterpart
$$|\hat{\beta}_i| = |\beta_i + h_i| \ge |\beta_i| +
\Real(\sgn(\beta_i)\overline{h_i}).$$ 
By carrying out these easy modifications (the details of which are
left to the reader) we can readily establish~\eqref{supportbound}
analogous to (1.11) of Theorem~1.3 in~\cite{CP08}.

Once we have recovered the support of $\vx$, call it $I$, we can solve 
for the coefficients of $\vx$ by solving the standard least
squares problem $\min \|\vA_I \vx_I - \vy\|_2$, where $\vA_I$ is tbe
submatrix of $\vA$ whose columns correspond to the support set $I$,
and similarly for $\vx_I$.
Statement~\eqref{condbound} follows by noting that the proof
of Theorem~3.2 in~\cite{CP08} yields as side result that 
with high probability the eigenvalues of any submatrix $\vA_I^{\ast} \vA_I$ 
with $|I| \le K$ are contained in the interval $[1/2, 3/2]$, 
which of course implies that $\kappa(\vA_I) \le \sqrt{3}$. The statement 
follows now by substituting this bound into the standard error 
bound, eq.~(5.8.11) in~\cite{HJ90}.
\end{proof}

\section*{Acknowledgements} \label{s:ack}
T.S.\ wants to thank Sasha Soshnikov for helpful discussions on random matrix
theory and Haichao Wang for a careful reading of the manuscript.


\end{document}